\documentclass[11pt]{article}
\usepackage{snapshot}
\usepackage{fullpage}
\usepackage{color}
\usepackage{tikz}
\usepackage{subcaption}
\usepackage{booktabs}
\usepackage[linesnumbered,ruled,vlined]{algorithm2e}
\usepackage{algpseudocode}
\SetKwInOut{Input}{Input}
\SetKwInOut{Output}{Output}
\usepackage{amsfonts}
\usepackage{amssymb}
\usepackage{amsthm}
\usepackage{amsmath}

\usepackage{ifpdf}
\newcommand{\mydriver}{hypertex}
\ifpdf
\renewcommand{\mydriver}{pdftex}
\fi
\usepackage[breaklinks,\mydriver]{hyperref}

\usepackage[capitalise]{cleveref}
\usepackage{footnote}

\makesavenoteenv{tabular}
\makesavenoteenv{algorithm}

\newcommand{\II}{\mathcal{I}}
\newcommand{\OPT}{\mathrm{OPT}}

\newcommand{\LL}{{\bf \mathcal{L}}}
\newcommand{\Poi}{\mathrm{Poi}}
\newcommand{\Var}{\mathrm{Var}}

\newcommand{\out}{\mathrm{out}}

\newcommand{\E}{\mathop{\mathbf{E}}}

\newcommand{\poly}{\mathrm{poly}}
\newcommand{\p}{\mathbf{p}}

\newcommand{\q}{\mathbf{q}}

\newcommand{\D}{\mathcal{D}}

\newcommand{\x}{\mathbf{x}}
\newcommand{\y}{\mathbf{y}}
\newcommand{\bc}{\mathbf{c}}
\newcommand{\z}{\mathbf{z}}
\newcommand{\vv}{\mathbf{v}}
\newcommand{\1}{\mathbf{1}}

\newcommand{\MC}{\mathrm{MC}}

\newtheorem{theorem}{Theorem}[section]
\newtheorem{fact}[theorem]{Fact}
\newtheorem{lemma}[theorem]{Lemma}
\newtheorem{corollary}[theorem]{Corollary}
\newtheorem{claim}[theorem]{Claim}

\theoremstyle{definition}
\newtheorem{definition}[theorem]{Definition}

\providecommand{\abs}[1]{\left\lvert#1\right\rvert}
\providecommand{\norm}[1]{\left\lVert#1\right\rVert}

\title{Sublinear-Time Algorithms for Max Cut, Max \texorpdfstring{E2Lin$(q)$}{Max E2Lin(q)}, and Unique Label Cover on Expanders}
\author{
	Pan Peng\footnote{Supported in part by NSFC grant 62272431 and ``the Fundamental Research Funds for the Central Universities''.}\\
	University of Science and Technology of China\\
	\texttt{ppeng@ustc.edu.cn}
	\and
	Yuichi Yoshida\footnote{Supported in part by JSPS KAKENHI Grant Number JP20H05965 and JP22H05001.}\\
	National Institute of Informatics\\
	\texttt{yyoshida@nii.ac.jp}
}

\begin{document}
\date{}
\maketitle
\begin{abstract}
	We show sublinear-time algorithms for \textsc{Max Cut} and \textsc{Max E2Lin$(q)$} on expanders in the adjacency list model that distinguishes instances with the optimal value more than $1-\varepsilon$ from those with the optimal value less than $1-\rho$ for $\rho \gg \varepsilon$.
  The time complexities for \textsc{Max Cut} and \textsc{Max $2$Lin$(q)$} are $\tilde{O}(\frac{1}{\phi^2\rho} \cdot m^{1/2+O(\varepsilon/(\phi^2\rho))})$ and $\widetilde{O}(\poly(\frac{q}{\phi\rho})\cdot {(mq)}^{1/2+O(q^6\varepsilon/\phi^2\rho^2)})$, respectively, where $m$ is the number of edges in the underlying graph and $\phi$ is its conductance.
	Then, we show a sublinear-time algorithm for \textsc{Unique Label Cover} on expanders with $\phi \gg \epsilon$ in the bounded-degree model.
  The time complexity of our algorithm is
  $\widetilde{O}_d(2^{q^{O(1)}\cdot\phi^{1/q}\cdot \varepsilon^{-1/2}}\cdot n^{1/2+q^{O(q)}\cdot \varepsilon^{4^{1.5-q}}\cdot \phi^{-2}})$,
  where $n$ is the number of variables.
  We complement these algorithmic results by showing that testing $3$-colorability requires $\Omega(n)$ queries even on expanders.
\end{abstract}

\thispagestyle{empty}
\setcounter{page}{0}
\newpage


\section{Introduction}

The max cut problem (\textsc{Max Cut}) is fundamental and has many applications in many areas of computer science. In \textsc{Max Cut}, given a graph $G=(V,E)$, we want to compute a bipartition $(V_1,V_2)$ of $V$ that maximizes the number of edges cut by the partition.
Let $\MC(G) \in [1/2,1]$ be the maximum fraction of the edges cut by a bipartition.
It is NP-hard to approximate to within a factor of $16/17\approx 0.941$~\cite{TSSW00}, but there is a 0.878-approximation algorithm based on SDP~\cite{GW95}, which is  tight assuming the unique games conjecture~\cite{Khot:2007bn}.

In this work, we consider sublinear-time algorithms for \textsc{Max Cut} in the \emph{adjacency list model}.
In this model, an algorithm can perform neighbor queries, i.e., for the $i$-th neighbor of any vertex, and degree queries, i.e., for the degree of any vertex.
We always assume that every vertex is incident to at least one edge to avoid technical trivialities.
In particular, the maximum cut size is $\Omega(n)$, where $n$ is the number of vertices.

The first contribution of this work is the first sublinear-time algorithm for \textsc{Max Cut} with non-trivial approximation guarantee for a natural class of graphs, i.e., expander graphs. To describe our results, we need several definitions.
Let $G=(V,E)$ be a graph with $n$ vertices and $m$ edges, and let $S \subseteq V$ be a vertex set.
The \emph{volume} and \emph{conductance} of $S$ are defined as $\mu_G(S) := \sum_{v \in S}d_G(v)$ and $\phi_G(S) := e_G(S,V \setminus S) / \mu_G(S)$, respectively, where $d_G(v)$ is the degree of $v$ and $e_G(S,V\setminus S)$ is the number of edges between $S$ and $V \setminus S$.
The volume and conductance of $G$ are defined to be $\mu_G := \sum_{v \in V}d_G(v)=2m$ and $\phi_G := \min_{\emptyset \subsetneq S:  \mu_G(S)\leq \frac{\mu_G}{2}}\phi_G(S)$, respectively.
We informally say that $G$ is an \emph{expander} when $\phi_G$ is bounded from below by a constant.

\begin{theorem}\label{thm:max-cut}
	There exists an algorithm that, given $\phi > 0$, query access to a graph $G=(V,E)$ with $\phi_G\geq \phi$ in the adjacency list model, and $\varepsilon,\rho>0$ with $\rho = \Omega(\varepsilon/\phi^2)$,
	\begin{description}
    \itemsep=0pt
		\item[(Completeness)] accepts $G$ with probability at least $2/3$ if $\MC(G) > 1-\varepsilon$,
		\item[(Soundness)] rejects $G$ with probability at least $2/3$ if $\MC(G) < 1-\rho$.
	\end{description}
	The time and query complexities of the algorithm are $\tilde{O}\left(\frac{m^{1/2+O(\varepsilon/(\phi^2\rho))}}{\phi^2\rho}\right)$, where $m$ is the number of edges in $G$.\footnote{$\tilde{O}(\cdot)$ hides polylogarithmic factors in $m$.}.
\end{theorem}
We note that $m$ does not have to be given as a part of the input.
By setting $\rho = 1/2-\varepsilon$, we can obtain an approximation ratio slightly better than the trivial approximation ratio of $1/2$:
\begin{corollary}\label{cor:max-cut}
	There exists an algorithm that, given $\phi > 0$, query access to a graph $G=(V,E)$ with $m$ edges and $\phi_G\geq \phi$ in the adjacency list model, and $\varepsilon = O(\phi^2)$,
	outputs a $(1/2+\varepsilon)$-approximation to $\MC(G)$.
	The time and query complexities of the algorithm are $\tilde{O}\left(\frac{m^{1/2+O(\varepsilon/\phi^2)}}{\phi^2}\right)$.
\end{corollary}

The only known sublinear-time algorithm in the adjacency list model related to \textsc{Max Cut} comes from a \emph{property testing} algorithm for bipartiteness, which distinguishes graphs with $\MC(G)=1$ from those with $\MC(G) < 1-\varepsilon$ with query complexity $\widetilde{O}(\sqrt{n}/\varepsilon^{O(1)})$~\cite{GR98:bipartite,Kaufman:2004vg}, assuming that we also have query access to the adjacency matrix of the input graph in addition to neighbor and degree queries,
and there is nearly matching lower bound $\Omega(\sqrt{n})$ on the number of queries for any testing algorithm of bipartiteness in this model~\cite{GR02:property} (see below for more discussions).

Bogdanov~et~al.~\cite{BOT02:test} showed that there is no sublinear-time algorithm for \textsc{Max Cut} problem with approximation ratio better than $16/17$, even in the \emph{bounded-degree model}, where the maximum degree of the underlying graph is bounded by a constant.
It is also known that~\cite{CKKMP18:cluster,Yos11:CSP}, for any $\varepsilon > 0$, any algorithm that approximates the maximum cut size of an $n$-vertex graph $G$ within a multiplicative error $1/2+\varepsilon$ must make at least $n^{1/2+\Omega(\varepsilon/\log(1/\varepsilon))}$ queries even when the underlying graph is a bounded-degree expander.
This lower bounds shows that the query complexity of Corollary~\ref{cor:max-cut} is almost tight because, if the underlying graph is a bounded-degree expander, i.e., $\phi_G = \Omega(1)$, then the query complexity of Corollary~\ref{cor:max-cut} becomes $m^{1/2+O(\varepsilon)} = n^{1/2+O(\varepsilon)}$.

We complement our algorithmic result for \textsc{Max Cut} by showing that testing $3$-colorability requires $\Omega(n)$ queries even on expanders. 
A $d$-bounded graph $G$ is said to be \emph{$\varepsilon$-far} from being $3$-colorable if one needs to remove at least $\varepsilon d n$ edges from $G$ to make it $3$-colorable. 
We have the following result.
\begin{theorem}\label{thm:3-colorability}
Let $G$ be a $d$-bounded graph $G$ with $\phi_G\geq \phi$ for some constant $\phi>0$. 
Any algorithm that distinguishes if $G$ is $3$-colorable or $\varepsilon$-far from being $3$-colorable with probability at least $2/3$ requires $\Omega(n)$ queries.
\end{theorem}

We note that sublinear-time algorithms for non-expanding graphs have been intensively studied. 
For example, Czumaj and Sohler characterized constant-query testable properties for planar graphs with one-sided error~\cite{czumaj2019characterization}.
Also, Newman and Sohler showed that every property is constant-query testable on bounded-degree planar graphs~\cite{newman2013every}, and then the query complexity was improved to $\exp(O(\epsilon^{-2}))$, which is tight~\cite{basu2022complexity}.
These results hold for more general hyperfinite graph classes, which include all minor-free graph classes.
By contrast, sublinear-time algorithms for expanding graphs are largely unexplored, and we take a step forward towards characterizing problems that can be solved in sublinear time on expanding graphs.

We next consider a more general problem \textsc{Max E2Lin$(q)$}.
An instance of \textsc{Max E2Lin$(q)$} is a tuple $\mathcal{I} = (G,q,\bc)$, where $G=(V,E)$ is a graph, $q \in \mathbb{Z}$ is a positive integer, and $\bc = \{\bc_e \in \mathbb{Z}_q \mid e \in E\}$ is a set of offsets.
For an assignment $\psi\colon V \to \mathbb{Z}_q$, we say that $\psi$ \emph{satisfies} the constraint $(u,v) \in E$ if $\psi(u)-\psi(v) = \bc_{uv}$ (in $\mathbb{Z}_q$).
The goal of \textsc{Max E2Lin$(q)$} is to find an assignment $\psi\colon V \to \mathbb{Z}_q$ that maximizes the number of satisfied constraints.
Let $\mathrm{OPT}(\mathcal{I}) \in [0,1]$ be the maximum fraction of constraints that can be satisfied by an assignment. Note that \textsc{Max Cut} can be seen as \textsc{Max E2Lin$(q)$} with $q = 2$ and each $\bc_e \in \mathbb{Z}_q$ being $1$.

Another contribution of this work is a sublinear-time algorithm for \textsc{Max E2Lin$(q)$} problem on expanders.
Here, we slightly modify the adjacency list model so that we also obtain $\bc_{uv}$ when we ask for a neighbor of $u \in V$ and the oracle returns a vertex $v \in V$ (See Section~\ref{sec:preliminaries} for details).
\begin{theorem}\label{thm:max-2lin}
There is an algorithm that, given $\phi > 0$ and query access to an instance $\II=(G,q,\bc)$ of \textsc{Max E2Lin$(q)$} with $m$ constraints and $\phi_G \geq \phi$ in the adjacency list model, and $\varepsilon,\rho >0$ with $\rho = \Omega(q^3\sqrt{\varepsilon}/\phi)$,
\begin{description}
  \itemsep=0pt
	\item[(Completeness)] accepts with probability at least $2/3$ if $\OPT(\mathcal{I}) > 1-\varepsilon$,
	\item[(Soundness)] rejects with probability at least $2/3$ if $\OPT(\mathcal{I}) < 1 - \rho$.
\end{description}
The time and query complexities of the algorithm are $\widetilde{O}(\poly(\frac{q}{\phi\rho})\cdot {(mq)}^{1/2+O(q^6\varepsilon/\phi^2\rho^2)})$.
\end{theorem}
By setting $\rho = (q-1)/q-\varepsilon$, we can obtain an approximation ratio slightly better than the trivial approximation ratio of $1/q$:
\begin{corollary}
  There exists an algorithm that, given $\phi > 0$ and query access to an instance $\II=(G,q,\bc)$ of \textsc{Max E2Lin$(q)$} with $m$ constraints and $\phi_G \geq \phi$ in the adjacency list model, and $\varepsilon = O(\phi^2/q^6)$, outputs a $(1/q+\varepsilon)$-approximation to $\OPT(\mathcal{I})$.
	The time and query complexities of the algorithm are $\widetilde{O}(\poly(\frac{q}{\phi})\cdot {(mq)}^{1/2+O(q^6\varepsilon/\phi^2)})$.
\end{corollary}

Yoshida and Ito~\cite{Yoshida:2010ua} showed that approximating \textsc{Max E2Lin$(2)$} within a factor of $11/12+\varepsilon$ in the bounded-degree model requires $\Omega(n)$ queries. 

Finally, we consider the unique label cover problem (\textsc{Unique Label Cover}).
In this problem, the input is a tuple $\mathcal{I} = (G,q,\pi)$, where $G$ is a graph, $q \in \mathbb{Z}$ is a positive integer, and $\pi = \{\pi_e \mid e\in E\}$ is a set of permutations over $[q] := \{1,2,\ldots,q\}$.
For an assignment $\psi:V \to [q]$, we say that $\psi$ \emph{satisfies} the constraint $e = (u,v) \in E$  if $\psi(v) = \psi(\pi_e(u))$.
The goal of the problem is to find an assignment $\psi:V \to [q]$ that maximizes the number of satisfied constraints.
We let $\OPT(\mathcal{I})$ denote the maximum fraction of satisfied constraints over all possible assignments.

We show the following sublinear-time algorithm for \textsc{Unique Label Cover} on expanders.
We say that a graph $G$ is \emph{$d$-bounded} if its maximum degree is at most $d$. Here, we slightly modify the adjacency list model so that we also obtain the permutation $\pi_{uv}$ when we ask for a neighbor of $u \in V$ and the oracle returns a vertex $v \in V$.
\begin{theorem}\label{thm:unique-game}
	There is an algorithm that, given $\phi > 0$, query access to a unique label cover instance $\mathcal{I}=(G,q,\pi)$ in the adjacency list model such that $G$ is $d$-bounded for some constant $d$ and $\phi_G \geq \phi$, and $\varepsilon,\rho > 0$ with $\varepsilon = O\left(\left(\frac{\phi^2}{q^{100}}\right)^{4^{q-1}}\right)$ and $\rho=\Omega_d(q^{86q}\cdot\varepsilon^{4^{1-q}}/\phi^4)$\footnote{Throughout the paper we use the notation $O_{d}(\cdot)$ (resp. $\Omega_{d}(\cdot)$) to describe a function in the Big-Oh (resp. Big-Omega) notation assuming that $d$ is constant.
	},
	\begin{description}
		\itemsep=0pt
		\item[(Completeness)] accepts $\mathcal{I}$ with probability at least $1-1/n$ if $\OPT(\mathcal{I}) \geq 1-\varepsilon$,
		\item[(Soundness)] rejects $\mathcal{I}$ with probability at least $1-1/n$ if $\OPT(\mathcal{I}) \leq 1-\rho$.
	\end{description}
	The time and query complexities of the algorithm are $\widetilde{O}_d(2^{q^{O(1)}\cdot\phi^{1/q}\cdot \varepsilon^{-1/2}}\cdot n^{1/2+q^{O(q)}\cdot \varepsilon^{4^{1.5-q}}\cdot \phi^{-2}})$.
\end{theorem}
We mention that the algorithm of Theorem~\ref{thm:unique-game} is \emph{spectral} and does not use semidefinite programming (SDP).
To the best of our knowledge, there is no known polynomial-time spectral algorithm for \textsc{Unique Label Cover} on expanders though an SDP-based algorithm for \textsc{Unique Label Cover}~\cite{Arora:2008hp} and spectral algorithms for \textsc{Max E2Lin$(q)$}~\cite{kolla2011spectral,li2019hermitian} are known.

\subsection{Technical Overview}

For \textsc{Max Cut}, we analyze the behavior of \emph{lazy} random walks on the graph $G=(V,E)$.
A lazy random walk of a fixed length naturally defines a \emph{simple} random walk by ignoring the steps that we stay at the current vertex. We can use the parity of the length, called the hop-length, of this simple random walk to infer the \textsc{Max Cut} value (see~\cite{KS11:maxcut}).
Slightly more formally, for any fixed starting vertex $v \in V$ and a length value $\ell$, we let $\D_{v,e}$ (resp., $\D_{v,o}$) denote the distribution of the endpoint of an $\ell$-step lazy random walk from $v$, conditioned on the event that the hop-length of the walk is even (resp., odd).
Then the distribution $\D_{v,o}$ will be ``far from'' $\D_{v,e}$ if $\OPT(G)$ is large, while the two distributions will be ``close to'' each other if $\OPT(G)$ is small and $G$ is an expander.
This intuition can be quantified by relating the bipartiteness ratio to the expansion and max cut value, and considering a variant of the $\ell_2$-norm distance between $\D_{v,o}$ and $\D_{v,e}$. For the former, we show that if an expander has small max cut, then its bipartiteness ratio is large, which can then be used to show that $\D_{v,o}$ is ``close to'' $\D_{v,e}$ using the Cheeger-like inequalities for bipartiteness ratio; if the max cut is large, then we use a spectral analysis to show that the two distributions are far apart. For the latter, we slightly modify the existing algorithm for testing the closeness of two distributions to approximate the distribution distances in sublinear time. 

We further remark that a graph $G$ is bipartite iff $\OPT(G)=1$ and $G$ is $\varepsilon$-far from being bipartite iff $\OPT(G)<1-\varepsilon$.
Thus, the previous bipartiteness testing algorithm in~\cite{GR98:bipartite,Kaufman:2004vg} can distinguish if $\OPT(G)=1$ or $\OPT(G)<1-\varepsilon$.
On the other hand, the problem of testing bipartiteness (on expanders) is easier than our problem because it suffices to find an odd cycle from a graph that is far from being bipartite.

Turning to \textsc{Max E2Lin$(q)$}, we make use of the structure of the label-extended graph $G_\II = (V' = V \times \mathbb{Z}_q,E')$ of the input instance $\II=(G,q,\bc)$, where $E' = \{((u,i),(v,i + \bc_{uv})) \mid (u,v) \in E, i \in \mathbb{Z}_q\}$.
Suppose $\OPT(\II) > 1-\varepsilon$, and let $\psi: V \to \mathbb{Z}_q$ be an assignment that attains it.
Then, we consider a family of sets $S_i := \{(v,\psi(v)+i \pmod q) \mid v \in V\}\;(i \in \mathbb{Z}_q)$ in $G_{\mathcal{I}}$.
Note that $S_0,\ldots,S_{q-1}$ forms a partition such that each $S_i$ has volume $\mu_{G_\mathcal{I}}/q$ and conductance $O(\varepsilon)$ (see Section~\ref{sec:preliminaries} for definitions).
On the other hand, we show that if $\OPT(\II)$ is small and the underlying graph $G$ is an expander, then any vertex set of volume $\mu_{G_{\mathcal{I}}}/q$ has a high conductance.
Then, the higher-order Cheeger inequality~\cite{KLLOT13:improved,lee2014multiway} implies that the $q$-th smallest eigenvalue of the normalized Laplacian matrix of $G_\II$ is large.
Finally, we observe that the previous $k$-clusterability testing algorithm from~\cite{CKKMP18:cluster} can be used to distinguish these two cases (and here we take $k=q-1$).

We need a more delicate argument for \textsc{Unique Label Cover}.
An issue that occurs when trying to apply the algorithm for \textsc{Max E2Lin$(q)$} to \textsc{Unique Label Cover} is that, in the completeness case, although we can guarantee that there exists one vertex set of volume $\mu_{G_\mathcal{I}}/q$ with a small conductance, there may not exist a partition into $q$ sets with each having a small conductance, and hence we cannot directly apply the $k$-clusterability testing algorithm in~\cite{CKKMP18:cluster}.
To resolve this issue, we first observe that when $\varepsilon$ is sufficiently small, we can find a partition into $r$ sets for some $r \leq q$ such that each has a small (outer) conductance, a large inner conductance (i.e., the subgraph induced by the set has large conductance), and volume at least $\mu_{G_{\mathcal{I}}}/q$.
Then, we can obtain query access to (an approximation to) this partition by modifying a spectral clustering oracle recently developed in~\cite{gluch2021spectral}.
With the query access at hand, we can verify that there exists a set with a small conductance and volume $\mu_{G_{\mathcal{I}}}/q$.
On the other hand, when $\mathrm{OPT}(\mathcal{I}) < 1-\rho$, we can show that there exists no such a set. 
The outline of our algorithm is as follows.
For each $r\leq q$, we sample a small set of vertices and invoke the spectral clustering oracle with appropriate parameters to determine the cluster membership of all the sampled vertices. Then we can estimate the volume of each potential cluster $\widehat{C}_i$ by using the fraction of sampled vertices that are reported to belong to $\widehat{C}_i$. We can then check if there exists a cluster of volume $\mu_{G_{\mathcal{I}}}/q$ and of a small conductance by a simple sampling approach. The algorithm accepts if and only if such a set (or cluster) is found.

\subsection{Related Work}\label{sec:relatedwork}
Our work is related to a line of study on testing bipartiteness, expansion, and $k$-clusterability in the framework of property testing. In this framework, given a property $\Pi$ and query access to a graph $G$, the goal is to design an algorithm that distinguishes if $G$ satisfies the property $\Pi$ or is $\varepsilon$-far from satisfying $\Pi$ by making as few queries as possible.
In the bounded-degree graph model, the graph $G$ is assumed to have maximum degree at most $d$ for some constant $d$, the allowed queries are neighbor queries, and $G$ is said to be $\varepsilon$-far from satisfying the property if one needs to insert and/or delete more than $\varepsilon dn$ edges to make it satisfy the property while preserving the degree bound.
In the adjacency list model for general graphs, i.e., no bound on the maximum degree, the allowed queries are neighbor queries and degree queries, and $G$ is said to be $\varepsilon$-far from satisfying the property if one needs to insert and/or delete more than $\varepsilon |E(G)|$ edges to make it satisfy the property.
Sometimes, the algorithm is also allowed to perform vertex-pair queries, i.e., if there exists an edge between any two vertices, in addition to the neighbor queries and degree queries.
The corresponding model is called the \emph{general graph model}.

Testing bipartiteness has been studied both in the bounded-degree graph model~\cite{GR98:bipartite} and the general graph model~\cite{Kaufman:2004vg}.
The algorithm of~\cite{Kaufman:2004vg} can test bipartiteness both in the general graph model with $\poly(\varepsilon^{-1}\log n)\tilde{O}(\min\{\sqrt{n},n^2/m\})$ queries, and in the adjacency list model with $\tilde{O}(\sqrt{n}/\varepsilon^{O(1)})$ queries, which almost matches the lower bound $\Omega(\sqrt{n})$ in this model.
As we mentioned before, these algorithms also distinguish if $\MC(G)=1$ or $\MC(G)<1-\varepsilon$ with a sublinear number of queries.
The problem of testing expansion in the bounded-degree graph model has been studied in~\cite{CS10:expansion,GR00:expansion,KS11:expansion,nachmias2010testing}.
Currently, the best known algorithm can distinguish if a graph has expansion at least $\phi$ or is $\varepsilon$-far from having expansion at least $\mu\phi^2$ in $\tilde{O}(n^{0.5+\mu}/(\varepsilon \phi^2))$ queries for any $\mu>0$. Li and Peng~\cite{li2015testing} gave a sublinear-time algorithm for testing a related notion called small set expansion.
Czumaj~et~al.~\cite{CPS15:cluster} then studied the problem of testing $k$-clusterability which generalizes the problem of testing expansion in the bounded-degree graph model.
Chiplunkar~et~al.~\cite{CKKMP18:cluster} then considered a slightly different (but very related) notion of $k$-clusterability from the one in~\cite{CPS15:cluster}. 
For general graphs, they gave an algorithm with query complexity $\poly(\frac{k\log m}{\beta})\cdot m^{1/2+O(\phi_{\out}/\phi_{\mathrm{in}}^2)}$ that distinguishes the case that a input graph $G$ can be partitioned into at most $k$ clusters with inner conductance at least $\phi_{\mathrm{in}}$ from the case that the vertex set of $G$ can be partitioned into $k+1$ pairwise disjoint subsets $C_1,\ldots,C_{k+1}$ such that $\mu_G(C_i)\geq \frac{\beta \mu_G}{k+1}$, and $\phi_G(C_i)\leq \phi_{\out}$ for each $1\leq i\leq k+1$. Peng \cite{peng2020robust} and Gluch et al. \cite{gluch2021spectral} gave sublinear-time clustering oracles for answering the cluster membership query in a (noisy) well-clusterable graph. 
We note that these works also use random walks and our algorithms build on the techniques developed there.

\textsc{Max Cut} and \textsc{Max E2Lin$(q)$} are special cases of \emph{maximum constraint satisfaction problems} (Max CSPs), where we are given a set of variables $V$ and constraints $C$ imposed on them, and the goal is to find an assignment to variables that satisfies as many constraints as possible.
Yoshida~\cite{Yoshida:2011da} showed that, for every CSP, the best possible approximation ratio that can be obtained by a constant-time algorithm in the bounded-degree model is determined by its LP integrality gap.
Based on his result, Kun~et~al.~\cite{Kun:2012bu} gave a combinatorial characterization of constant-time testable CSPs, where the goal is to distinguish instances $\mathcal{I}$ with $\mathrm{OPT}(\mathcal{I}) = 1$ from those with $\mathrm{OPT}(\mathcal{I}) < 1 - \varepsilon$.

For a Boolean predicate $P\colon {\{0,1\}}^k \to \{0,1\}$, the instance of Max-CSP$(P)$ is a $k$-uniform hypergraph $G=(V,E)$, where the vertices in each hyperedge are ordered.
Then for an assignment $\psi: V\to \{0,1\}$, we say that a hyperedge (or a constraint) $(v_1,\ldots,v_k) \in E$ is \emph{satisfied} if $(\psi(v_1),\ldots,\psi(v_k)) \in P^{-1}(1)$.
The goal of the problem is to to find an assignment $\psi\colon V \to \{0,1\}$ that maximizes the number of satisfied constraints.
Yoshida~\cite{Yos11:CSP} studied lower bound for Max-CSP$(P)$ in the bounded-degree model.
First, a predicate $P\colon {\{0,1\}}^k \to \{0,1\}$ is called \emph{symmetric} if (i) $P(x) = P(y)$ for any $x,y \in {\{0,1\}}^k$ with $|x| = |y|$, and (ii) $P(x) = P(\bar{x})$ for any $x \in {\{0,1\}}^k$, where $\bar{x} = (1-x_1,\ldots,1-x_k)$.
A representative example of symmetric predicates is $\mathrm{EQU}\colon {\{0,1\}}^k \to \{0,1\}$, which outputs one iff the variables are all zeros or all ones, and $\mathrm{NAE}\colon {\{0,1\}}^k \to \{0,1\}$, which outputs one iff not all the variables have the same value. Then, Yoshida~\cite{Yos11:CSP} showed that, for a symmetric predicate $P\colon {\{0, 1\}}^k \to \{0, 1\}$ except EQU with $k \geq 3$, for any $\varepsilon > 0$, and predicate $Q \colon {\{0, 1\}}^k \to \{0, 1\}$ with $P^{-1}(1) \subseteq Q^{-1}(1)$,
we need $\Omega(n^{1/2+\delta(\varepsilon)} )$ queries for some $\delta(\varepsilon) > 0$ to distinguish satisfiable instances of Max-CSP$(Q)$ from those with optimum value at most $|Q^{-1}(1)|/2^k +\varepsilon$.

\subsection{Organization}
In Section~\ref{sec:preliminaries}, we introduce notions that will be used throughout this paper.
We prove Theorems~\ref{thm:max-cut},~\ref{thm:max-2lin},~\ref{thm:unique-game}, and~\ref{thm:3-colorability} in Sections~\ref{sec:max-cut},~\ref{sec:max-2lin},~\ref{sec:unique-label-cover}, and~\ref{sec:3color}, respectively.


\section{Preliminaries}\label{sec:preliminaries}
We will use bold letters for vectors.
For any vector $\p \in \mathbb{R}^V$ on a vertex set $V$ and a set $S \subseteq V$, let $\p(S):=\sum_{v\in S}\p(v)$. Let $\1_S$ be the indicator vector of $S$, i.e., $\1_S(u)=1$ if $u\in S$ and $0$ otherwise. Let $\1_{S,T}:=\1_S-\1_T$.

Let $G=(V,E)$ be a graph.
For a vertex $u \in V$, we denote by $N_G(u)$ the set of neighbors of $u$.
For a vertex subset $S\subseteq V$, let $\bar{S}:=V\backslash S$ denote the complement of $S$.
We define $G[S]$ as the subgraph of $G$ induced by $S$.
We define $e_G(S)$ as the total number of edges within $S$.
For any two subsets $L,R\subseteq V$, let $e_G(L,R)$ denote the total number of edges between $L$ and $R$.

Given a graph $G = (V, E)$, let $A_G \in \mathbb{R}^{V \times V}$ be its adjacency matrix, and let $D_G \in \mathbb{R}^{V \times V}$ be its diagonal degree matrix.
The \emph{Laplacian} and \emph{normalized Laplacian} matrices of $G$ are defined as $L_G:=D_G-A_G$ and $\mathcal{L}_G:=I-D_G^{-1/2}A_G D_G^{-1/2}$, respectively.
Let $0=\lambda_1(G) \leq \cdots \leq \lambda_{n}(G) \leq 2$ be the eigenvalues of the normalized Laplacian $\LL_G$.

When the graph $G$ is clear from the context, we omit $G$ from the notations above, and we use $n$ and $m$ to denote the number of vertices and edges, respectively, of $G$.

\subsection{Conductance and Cheeger inequalities}\label{subsec:conductance}

Let $G=(V,E)$ be a graph.
For $k \geq 1$, we define
$\phi_G(k):= \min\limits_{\emptyset \subsetneq S \subsetneq V: \mu_G(S)\leq \mu_G/k} \phi_G(S)$ and
$\rho_G(k):=\min\limits_{S_1,\ldots,S_k}\max\limits_{1\leq i\leq k}\phi_G(S_i)$,
where in the definition of $\rho_G(k)$, $S_1,\ldots, S_k\subseteq V$ are over all $k$ disjoint non-empty subsets of $V$.
Note that $\phi_G(2) = \phi_G$ and that $\phi_G(k)\leq \rho_G(k)$ for any $k\geq 2$,
The following Cheeger inequalities are well known.
\begin{lemma}[Cheeger inequality~\cite{Alon:1986gz,Alon:1985jg}]
  We have
    $\frac{\lambda_2}{2} \leq \phi_G \leq \sqrt{2\lambda_2}$.
\end{lemma}
\begin{theorem}[Higer-order Cheeger inequality~\cite{KLLOT13:improved,lee2014multiway}]\label{thm:higher-order-cheeger}
We have
  	$\frac{\lambda_k }{2}\leq \rho_G(k)\leq O(k^2 \sqrt{\lambda_k})$.
\end{theorem}

\subsection{Bipartiteness ratio}
Let $G=(V,E)$ be a graph and $L,R \subseteq V$ be two disjoint sets such that $L\cup R \neq \emptyset$.
The \textit{bipartiteness ratio of $(L,R)$} is defined as
\[
\beta_G(L,R):=\frac{2e_G(L)+2e_G(R)+e_G(L\cup R,\overline{L\cup R})}{\mu_G(L\cup R)}.
\]
The \textit{bipartiteness ratio of a nonempty set $S$} is defined to be the minimum value of $\beta(L,R)$ over all possible partitions $(L,R)$ of $S$, i.e.,
$$\beta_G(S):=\min_{(L,R):\,\textrm{partition of $S$}}\beta_G(L,R).$$
We note that the bipartiteness ratio is between $0$ and $1$, and if the bipartiteness ratio is small, then the subgraph induced by $S$ is highly bipartite and is well separated from the rest.
The bipartiteness ratio of the graph $G$ is defined as $\beta_G :=\min_{\emptyset \subsetneq S\subsetneq V}\beta_G(S)$.
The following Cheeger-like inequalities for bipartiteness ratio are known.
\begin{lemma}[\cite{Tre09:maxcut}]\label{lemma:tre}
  We have $\frac{2-\lambda_{n}}{2}\leq \beta_G \leq \sqrt{2(2-\lambda_{n})}$.
\end{lemma}
\begin{theorem}[Special case of Theorem 4.10 in~\cite{KLLOT13:arxiv}]\label{thm:kllot}
Let $1\leq k\leq n$. It holds that
$\beta_G \leq O\left(\frac{k\cdot (2-\lambda_n)}{\sqrt{2-\lambda_{n-k}}}\right)$.
\end{theorem}

We define the \emph{$k$-way dual Cheeger constant} as
\[
\bar{h}_G(k)=\max_{(V_1,V_2),\dots,(V_{2k-1},V_{2k})}\min_{1\leq i\leq k}\frac{2e(V_1,V_2)}{\mu_G(V_1\cup V_2)},
\] where the maximization is over all collections of $k$ pairs of subsets $(V_1,V_2),\dots,(V_{2k-1},V_{2k})$ which satisfy $V_{p}\cap V_q=\emptyset$, for all $1\leq p\neq q\leq 2k$, $V_{2i-1}\cup V_{2i}\neq \emptyset $, for all $1\leq i\leq k$.

\begin{theorem}[Theorem 1.2 in~\cite{liu2015multi}]\label{thm:liu}
	Let $1\leq k\leq n$. It holds that
	$\frac{2-\lambda_{n-k+1}}{2}\leq 1-\bar{h}_G(k)\leq O\left(k^3\sqrt{2-\lambda_{n-k+1}}\right)$.
\end{theorem}



\section{Max Cut on Expanders}\label{sec:max-cut}
In this section, we describe our sublinear-time algorithm for estimating the maximum cut size on expanders and prove Theorem~\ref{thm:max-cut}.
\subsection{Algorithm description}
The first component of our algorithm is lazy random walk.
Let $P:=D^{-1}A$ be the probability transition matrix of the following \textit{simple} random walk on $G$: At each step, we jump to neighbor chosen uniformly at random of the current vertex.
Let $W:=\frac{I+D^{-1}A}{2}=\frac{I+P}{2}$ be the probability transition matrix of the following \textit{lazy} random walk matrix of $G$: At each step, with probability $1/2$, we stay at the current vertex; and with the remaining probability $1/2$, we take a step of the simple random walk.
We can see that a lazy random walk of length $t$ is equivalent to the following process: first flip an unbiased coin $t$ times; suppose that the number of heads seen is $h$, then perform a simple random walk of length $h$. We define this number $h$, namely the number of steps at which a simple random walk is performed, to be the \textit{hop-length} of the original $t$-step lazy walk. And we call a lazy random walk is of \textit{even (reps., odd) hops} if $h$ is even (resp., odd). Such a view of lazy random walks and their hop-lengths have also been used by Kale and Seshadhri~\cite{KS11:maxcut}.

The second component of our algorithm is an efficient tester \Call{$l_2$-DifferenceTest}{} for the closeness of two vectors that are transformed versions of two distributions.
Here we identify a distribution over the vertex set $V$ with a vector $\mathbf{p} \in \mathbb{R}^V$ with $\mathbf{p}(u)\;(u \in V)$ being the probability that $u$ is sampled.
The guarantee of the tester is given in the following lemma, which is a modification of Theorem 1.2 from~\cite{CDVV14:optimal}, and we defer its proof to Appendix~\ref{app:deferred_proofs}.
\begin{theorem}\label{thm:distribution}
  Let $G=(V,E)$ be a graph, and let $\p,\q$ be distributions over $V$.
  There exists an algorithm \Call{$l_2$-DifferenceTest}{} that takes as input an integer $r \geq 1$ and sampling accesses to $\p$ and $\q$ such that for any $\delta, \xi > 0$,
  \begin{itemize}
    \itemsep=0pt
    \item accepts with probability at least $1-\delta$ if $\norm{(\p-\q) D_G^{-1/2}}_2^2\le \xi$,
    \item rejects with probability at least $1-\delta$ if $\norm{(\p-\q) D_G^{-1/2}}_2^2\ge 4\xi$,
  \end{itemize}
  provided that $r \ge C_{\ref{thm:distribution}} \cdot \frac{\sqrt{b}}{\xi} \ln\frac{1}{\delta}$, where $b = \max\left\{\norm{\p D_G^{-1/2}}_2^2, \norm{\q D_G^{-1/2}}_2^2\right\}$ and $C_{\ref{thm:distribution}} \geq 1$ is an absolute constant.
  The time and sample complexities are linear in $r$.
\end{theorem}
A pseudocode of our algorithm is given in Algorithm~\ref{alg:max-cut}.

\begin{algorithm}[t!]
  \caption{\textsc{TestExpanderMC}$(\phi,G,\varepsilon,\rho)$}\label{alg:max-cut}
  Sample $O(1)$ vertices with probability proportional to their degrees\;
  \For{each sampled vertex $v$}{
  Let $\ell:=\frac{1}{16C \phi^2\rho}\ln \mu_G$ for some constant $C=C_{\ref{thm:max-cut}}$\;
  Let $\D_{v,e}$ (resp., $\D_{v,o}$) denote the distribution of the endpoint of an $\ell$-step lazy random walk from $v$, conditioned on the event that the hop-length of the walk is even (resp., odd)\;
  Invoke \textsc{$l_2$-DifferenceTest} to test if $\norm{(\D_{v,e}-\D_{v,o})\cdot D^{-1/2}}_2^2\leq \xi_{\mathrm{trm}}:=\frac{1}{3600\mu_G^{1+\varepsilon/(2C \phi^2\rho)}}$\;
  \lIf{this tester rejects}{Abort and output \textbf{Accept}.}
  }
  Output \textbf{Reject}.
\end{algorithm}

\paragraph{Implementation of  Algorithm~\ref{alg:max-cut}.}

In our query model, to implement the first line (i.e., sampling a random vertex with probability proportional to its degree), we can make use of a procedure by Eden and Rosenbaum~\cite{ER18:sampling}.
In particular, using $\tilde{O}(n/\sqrt{\eta \mu_G})$ degree and neighbor queries, the procedure can produce a vertex $v$ that is sampled with probability $(1\pm \eta)d(v)/\mu_G$, for any small constant $\eta>0$.
Since we only need to sample $O(1)$ vertices with probability proportional to their degrees, it will only incur $\tilde{O}(n/\sqrt{\eta \mu_G})$ additional queries, which are dominated by the query and time complexities for invoking \Call{$l_2$-DifferenceTest}{}.
Furthermore, we need to approximate the volume $\mu_G$ (i.e., $2m$), within a small constant factor, of the graph, in order to set our parameters $\ell$ and $\xi_{\mathrm{trm}}$.
This approximation can be done by one algorithm of Goldreich and Ron~\cite{GR08:approximating}, who gave an algorithm that uses $\tilde{O}(n/\sqrt{m})$ degree and neighbor queries in expectation and outputs an estimate $\hat{m}$ that is within a factor of $2$ of $m$ with probability at least $2/3$.
To make the exposition simpler, we assume that we can sample a vertex $v$ with probability $d(v)/\mu_G$ and we know the volume $\mu_G$.

Finally, in order to perform the lazy random walk at any vertex $v$, it suffices to have the degree of $v$ and sample a random neighbor $u$. The subroutine \textsc{$l_2$-DifferenceTest} requires to query the degree of vertices, and get some samples from the two distributions $\D_{v,e}$ and $\D_{v,o}$, and each such sample can be obtained by performing a random walk of length $\ell$ from $v$ and then taking the endpoint.

\subsection{Proof of Theorem~\ref{thm:max-cut}}
\subsubsection{Key properties of collision probabilities}
For two vertices $u,v \in V$ and an integer $t \geq 0$, let $\p_v^t(u)$ be the probability that a $t$-step lazy random walk starting from $v$ ends at $u$.
Let $\p_{v,e}^t(u)$ (resp., $\p_{v,o}^t(u)$) be the probability that a $t$-step lazy random walk of even (resp., odd) hop-length starting from $v$ ends at $u$. Note that
\begin{align}
\p_{v,e}^t(u)+\p_{v,o}^t(u)=\p_v^t(u)=\1_v W^t (u).\label{eqn:p_even_odd}
\end{align}

Now let $M:=I-W=\frac{I-D^{-1}A}{2}$. Let $\overline{M}:=\frac{I-D^{-1/2}AD^{-1/2}}{2}=\frac{\LL}{2}$. Note that $M=D^{-1/2}\overline{M}D^{1/2}$; and that for any $1\leq i\leq n$, the $i$-th (left) eigenvector and $i$-th eigenvalue of $\overline{M}$ are $\vv_i$ and $\frac{\lambda_i}{2}$, respectively. Furthermore, the $i$-th eigenvector and $i$-th eigenvalue of ${M}$ are $\vv_i D^{1/2}$ and $\frac{\lambda_i}{2}$, respectively. The $i$-th eigenvector and $i$-th eigenvalue of ${W}$ are $\vv_i D^{1/2}$ and $1-\frac{\lambda_i}{2}$, respectively.

We will use $M$ to characterize the parity of the hop-length of the lazy random walks.
Let $\q_{v}^t:=\1_v M^t$. Note that $\q_v^t$ does \emph{not} represent a probability distribution.
We first give the following simple observations:
\begin{fact}\label{fact:p_q_vectors}
  For any vertices $u,v \in V$ and integer $t \geq 0$, $\q_v^t(u)=\p_{v,e}^t(u)-\p_{v,o}^t(u)$.
\end{fact}
\begin{proof}
  By definition,
  \begin{align*}
  \q_v^t(u) =\1_v M^t(u)=\sum_{i: \mathrm{even}}\binom{t}{i}{\left(\frac{1}{2}\right)}^{t-i}\1_v{\left(\frac{P}{2}\right)}^{i}(u)-\sum_{i: \mathrm{odd}}\binom{t}{i}{\left(\frac{1}{2}\right)}^{t-i}\1_v{\left(\frac{P}{2}\right)}^{i}(u)
  =\p_{v,e}^t(u)-\p_{v,o}^t(u),
  \end{align*}
  as desired.
\end{proof}

\begin{fact}~\label{fact:el2norm}
  For any vertex $v \in V$ and integer $t \geq 0$, $\sum_{u \in V} \p_{v,e}^t(u)=\sum_{u \in V}\p_{v,o}^t(u)=1/2$.
\end{fact}
\begin{proof}
By definition,
  \begin{align*}
  & \sum_{u \in V}\p_{v,e}^t(u)= \sum_{u \in V}\sum_{i:\mathrm{even}}\binom{t}{i}{\left(\frac{1}{2}\right)}^{t-i}\1_v{\left(\frac{P}{2}\right)}^{i}(u) \\
  & =\frac{1}{2^t}\sum_{i:\mathrm{even}}\binom{t}{i}\sum_{u \in V}\1_v P^{i}(u)
  =\frac{1}{2^t}\sum_{i:\mathrm{even}}\binom{t}{i}\cdot 1
  =\frac{2^{t-1}}{2^t}=\frac{1}{2}.
  \end{align*}
  By similar calculation, we can also show that $\sum_{u \in V}\p_{v,o}^t(u)=\frac{1}{2}$.
\end{proof}

For any integer $t\geq 0$, we let
$$\Delta_{t}(v):=\norm{\q_{v}^t D^{-1/2}}_2^2=\norm{(\p_{v,e}^t-\p_{v,o}^t) D^{-1/2}}_2^2.$$
We will bound $\Delta_{t}(v)$ to analyze our algorithm.

Note that if we restrict $t$ to be $\ell$, then by definitions of $\D_{v,e}$ and $\D_{v,o}$ in our algorithm and the above fact, it holds that
$$\D_{v,e}=2\p_{v,e}^\ell, \quad \D_{v,o}=2\p_{v,o}^\ell, \quad \Delta_{\ell}(v)= \frac{1}{4}\norm{(\D_{v,e}-\D_{v,o}) D^{-1/2}}_2^2.$$

We have the following simple lemma.
\begin{lemma}\label{lemma:expander}
  Let $G$ be a graph with $\phi_G\geq \phi$.
  Then
  \begin{align*}
    \norm{\p_v^t D^{-1/2}}_2^2\leq \frac{1}{\mu_G}+{\left(1-\frac{\phi^2}{4}\right)}^{2t}.
  \end{align*}
\end{lemma}
\begin{proof}
  Let $\vv_i \in \mathbb{R}^V$ and $\lambda_i \in \mathbb{R}$ be the $i$-th eigenvector and eigenvalue of $\mathcal{L}$.
  Since $\phi_G\geq \phi$, we have $\lambda_2\geq \frac{\phi^2}{2}$ by Cheeger's inequality.
  By writing $\1_v D^{-1/2}:=\sum_i\alpha_i \vv_i$, we have
  \begin{align*}
    \p_v^t D^{-1/2}&=\1_v W^t D^{-1/2} = \1_v {\left(I-D^{-1/2}\overline{M}D^{1/2}\right)}^t D^{-1/2} = \1_v D^{-1/2} {(I-\overline{M})}^t \\
    &=\left(\sum_i\alpha_i \vv_i\right) \left(\sum_i{\left(1-\frac{\lambda_i}{2}\right)}^t \vv_i^T\vv_i\right)  = \sum_{i}\alpha_i{\left(1-\frac{\lambda_i}{2}\right)}^t \vv_i
  \end{align*}
  Thus, we have that
  \[
    \norm{\p_v^t D^{-1/2}}_2^2 =\sum_{i=1}^{n}\alpha_i^2{\left(1-\frac{\lambda_i}{2}\right)}^{2t}.
  \]
  Note that the first eigenvector of the normalized Laplacian $\LL$ is $\vv_1={\left(\sqrt{d(v)/\mu_G}\right)}_{v \in V}$, and thus $\alpha_1=\langle\1_v D^{-1/2},\vv_1\rangle = \sqrt{1/\mu_G}$.
  Furthermore, $\sum_{i=2}^{n}\alpha_i^2\leq \sum_{i=1}^{n}\alpha_i^2=\norm{\1_v D^{-1/2}}_2^2=1/d(v)\leq 1$ as we have assumed that every vertex has degree at least one.
  Thus
  \[
    \norm{\p_v^t D^{-1/2}}_2^2=\frac{1}{\mu_G}+\sum_{i=2}^{n}\alpha_i^2{\left(1-\frac{\lambda_i}{2}\right)}^{2t}\leq \frac{1}{\mu_G}+{\left(1-\frac{\phi^2}{4}\right)}^{2t}. \qedhere
  \]
\end{proof}

\subsubsection{Soundness}
In this section, we prove the soundness of our algorithm, and in the rest of this section, we assume that the input graph $G$ satisfies $\MC(G)\leq 1-\rho$ and $\phi_G\geq \phi$ for $0<\phi<1$ and $0<\rho\leq 1/2$.
We first show that such a graph $G$ has a large bipartiteness ratio.
\begin{lemma}\label{lem:bipartiteness-ratio-lower-bound}
  We have $$\beta_G \geq \frac{\phi\rho}{2}.$$
\end{lemma}
\begin{proof}
  Let $S\subseteq V$. We consider the following three cases:
  \begin{enumerate}
  \item Suppose $\mu(S)\leq \mu_G/2$.
  Then $e(S,\overline{S})\geq \phi \cdot \mu(S)$, and thus
  $$\beta(S)\geq \frac{e(S,\overline{S})}{\mu(S)}\geq \phi.$$
  \item Let $0<x< 1$ be specified later.
  Suppose $\mu_G/2<\mu(S)\leq (2-x)\mu_G/2$.
  Then $e(S,\overline{S})\geq \phi \cdot \mu(\overline{S})\geq \phi x \mu_G$, and thus
  $$\beta(S)\geq \frac{e(S,\overline{S})}{\mu(S)}\geq \frac{\phi x \mu_G/2}{(2-x)\mu_G/2}=\frac{\phi x}{2-x}.$$
  \item Suppose $\mu(S)>(2-x)\mu_G/2$.
  Since $\MC(G)\leq 1-\rho$, it holds that for any partition $(L,R)$ of $S$, $e(L,R)\leq \MC(G)\cdot \mu_G/2\leq (1-\rho) \mu_G/2$.
  Thus,
  \begin{align*}
  \beta(L,R)=1-\frac{2e(L,R)}{\mu(S)}\geq 1-\frac{2(1-\rho) \mu_G/2}{(2-x)\mu_G/2}=1-\frac{2-2\rho}{2-x}\\
  =\frac{2\rho-x}{2-x}.
  \end{align*}
  \end{enumerate}
  Therefore, $\beta_G \geq \min\{\phi, \frac{\phi x}{2-x}, \frac{2\rho-x}{2-x} \}=\min\{\frac{\phi x}{2-x}, \frac{2\rho-x}{2-x}\}$.
  By setting $x=\frac{2\rho}{1+\phi}$, we get that $\beta_G \geq \frac{2\phi\rho}{2+2\phi-2\rho}\geq \frac{\phi\rho}{2}$, where the last inequality follows from the fact that $2+2\phi-2\rho\leq 4$.
  This completes the proof of the lemma.
\end{proof}

\begin{lemma}\label{lemma:betaG_stronger}
  We have $$\beta_G = O\left(\frac{2-\lambda_n}{\phi}\right).$$
\end{lemma}
\begin{proof}
  We first show that $1-\bar{h}_G(2)\geq \phi$. This is true as for any two disjoint subsets $S_1,S_2 \subseteq V$, there must exist one subset, say $S_1$, with volume at most $\mu_G/2$. Then we observe that
  $$\beta(S_1)\geq \frac{e(S_1,\overline{S_1})}{\mu(S_1)}=\phi(S_1)\geq \phi.$$
  Thus, $\max\{\beta(S_1),\beta(S_2)\}\geq \phi$. Let $L_1,R_1$ (resp. $L_2,R_2$) be any partition of $S_1$ (resp. $S_2$). Then $1-\frac{2e(L_1,R_1)}{\mu_G(L_1\cup R_1)}=\beta(L_1,R_1)\geq \beta(S_1)$, and $1-\frac{2e(L_2,R_2)}{\mu_G(L_2\cup R_2)}=\beta(L_2,R_2)\geq \beta(S_2)$. 
Then 
\[
\max\{ 1-\frac{2e(L_1,R_1)}{\mu_G(L_1\cup R_1)}, 1-\frac{2e(L_2,R_2)}{\mu_G(L_2\cup R_2)}\}\geq \max\{\beta(S_1),\beta(S_2)\}\geq \phi. 
\]
Thus,
\[
\min\{ \frac{2e(L_1,R_1)}{\mu_G(L_1\cup R_1)}, \frac{2e(L_2,R_2)}{\mu_G(L_2\cup R_2)}\}\leq 1-\phi. 
\]
  
Since this inequality holds for any two subsets $(L_1,R_1),(L_2,R_2)$ such that $L_1\cap L_2=R_1\cap R_2=\emptyset, L_i\cap R_i=\emptyset$ and $L_i\cup R_i\neq \emptyset$, for $i=1,2$, we have that $\bar{h}_G(2)\leq 1-\phi$.

  Now by Theorem~\ref{thm:liu}, we have $\sqrt{2-\lambda_{n-1}} = \Omega(1- \bar{h}_G(2)) = \Omega(\phi)$.
  Finally, we apply Theorem~\ref{thm:kllot} to get $\beta_G = O\left(\frac{2-\lambda_n}{\phi}\right)$, which finishes the proof of the lemma.
\end{proof}

\begin{lemma}\label{lemma:soundness}
  For any starting vertex $v \in V$, we have
  \begin{align}
    \Delta_\ell(v)=\norm{\1_v M^\ell \cdot D^{-1/2}}_2^2\leq \frac{\xi_{\mathrm{trm}}}{4}\label{eqn:soundness}
  \end{align}
\end{lemma}
\begin{proof}
  Note that ${\{\lambda_i/2\}}_{1\leq i\leq n}$ are the eigenvalues of $M$ (Recall that $0=\lambda_1\leq\lambda_1\leq\cdots\leq\lambda_{n}\leq 2$ are the eigenvalues of the normalized Laplacian $\LL$).

  By Lemma~\ref{lem:bipartiteness-ratio-lower-bound}, it holds that $\beta_G \geq \frac{\phi\rho}{2}$, and by Lemma~\ref{lemma:betaG_stronger}, $\beta_G = O(\frac{2-\lambda_n}{\phi})$. Thus
  $$\frac{\lambda_{n}}{2} = 1-\Omega({\beta_G }\phi) \leq 1-C_0 \phi^2\rho,$$
  for some constant $C_0>0$ and we set $C_{\ref{thm:max-cut}}:=\frac{C_0}{16}$.

If we let $\1_v D^{-1/2}=\sum_{i=1}^n \alpha_i \vv_i$, then for every vertex $v \in V$ and any $t\geq 1$, we have that
  \begin{align*}
    \Delta_t(v) & =\norm{\1_v M^t \cdot D^{-1/2}}_2^2 = \norm{\1_v D^{-1/2}{(I-\overline{M})}^t}
    = \sum_{i=1}^n \alpha_i^2 {\left(\frac{\lambda_i}{2}\right)}^{2t}
    \leq {\left(\frac{\lambda_{n}}{2}\right)}^{2t}
    \leq  {\left(1-C_0 \phi^2\rho \right)}^{2t},
  \end{align*}
where in the second to last inequality, we used that fact that $\sum_{i=1}^n \alpha_i^2=\norm{\1_v\cdot D^{-1/2}}_2^2=\frac{1}{d(v)}\leq 1$.

  Furthermore, recall that $\rho \geq \frac{\varepsilon}{C_{\ref{thm:max-cut}}\phi^2}=\frac{16\varepsilon}{C_0\phi^2}$ and that $\ell=x\ln \mu_G$ for $x=\frac{1}{16C_{\ref{thm:max-cut}} \phi^2\rho}=\frac{1}{C_0 \phi^2\rho}$,
  which give
  \begin{align*}
  \Delta_\ell(v)\leq  {\left(1-C_0 \phi^2\rho \right)}^{2\ell}\leq e^{-2x C_0 \phi^2\rho\ln \mu_G}
 =  \frac{1}{\mu_G^{2xC_0\phi^2\rho}} = \frac{1}{\mu_G^{2}}\leq \frac{1}{14400 \mu_G^{1+8x\varepsilon}}=\frac{\xi_{\mathrm{trm}}}{4},
  \end{align*}
where in the last inequality, we used the fact that $\mu_G\geq n$, which is sufficiently large; and the fact that $\varepsilon x = \frac{\varepsilon}{C_0 \phi^2\rho}<\frac{1}{16}$.
\end{proof}

\subsubsection{Completeness}
Now we show that if the input graph $G$ satisfies $\MC(G)\geq 1-\varepsilon$, then the algorithm will accept $G$ with probability at least $2/3$. We will need the following lemma.
In the rest of this section, we fix $G$ to be a graph with $\MC(G)\geq 1-\varepsilon$.

\begin{lemma}~\label{lem:maxcut}
  If there exists a partition $(S,\overline{S})$ with cut value at least $1-\varepsilon$ with $\varepsilon<1/2$, then for any integer $t\geq 0$, there exists a subset $V_g\subseteq V$ such that $\mu(V_g)\geq \mu_G/8$, and that for any $v\in V_g$, we have
  \begin{align}
    \sum_{u\in V}|\q_{v}^t(u)|\geq \frac{1}{60}{(1-2\varepsilon)}^t. \label{eqn:l1lower}
  \end{align}
  Therefore,
  \begin{align}
    \Delta_t(v)&=\sum_{u\in V}\left|\q_{v}^t(u)\cdot \frac{1}{\sqrt{d(u)}}\right|^2 =\sum_{u\in V}{\left|\q_{v}^t(u)\cdot \frac{1}{\sqrt{d(u)}}\right|}^2 \cdot \sum_{u\in V} {\left(\sqrt{d(u)}\right)}^2 \cdot \frac{1}{\mu_G} \nonumber \\
    &\geq\frac{1}{\mu_G}{\left(\sum_{u\in V}|\q_{v}^t(u)|\right)}^2\geq \frac{1}{3600\mu_G}{(1-2\varepsilon)}^{2t}. \label{eqn:l2lower}
  \end{align}
\end{lemma}

We first use this lemma to prove a key property for the completeness part.

\begin{lemma}\label{lemma:completeness}
  For any $v\in V_g$ as defined in Lemma~\ref{lem:maxcut}, we have
  \[
  \Delta_\ell(v)=\norm{\1_v M^\ell \cdot D^{-1/2}}_2^2\geq \xi_{\mathrm{trm}}\cdot \frac{}{}\label{eqn:completeness}
  \]
\end{lemma}
\begin{proof}
  Let $v\in V_g$.
  Recall that $\rho \geq \frac{\varepsilon}{C_{\ref{thm:max-cut}} \phi^2}$, that $\ell=x\ln \mu_G$ for $x=\frac{1}{16C_{\ref{thm:max-cut}} \phi^2\rho}$.
  We have
  \[
    \Delta_{\ell}(v)\geq \frac{1}{3600\mu_G}{(1-2\varepsilon)}^{2\ell}\geq \frac{1}{3600\mu_G}\cdot e^{-8x\varepsilon \ln \mu_G} =\frac{1}{3600\mu_G^{1+8x\varepsilon}}=\xi_{\mathrm{trm}}.
    \qedhere
  \]
\end{proof}

Now we give the proof of Lemma~\ref{lem:maxcut}.

\begin{proof}[Proof of Lemma~\ref{lem:maxcut}]
First note that by the precondition of the lemma, it holds that $e(S,\overline{S})\geq (1-\varepsilon) m$.

Now we introduce some notations. For any vertex subset $U\subseteq V$, we let $U_S:=U\cap S$ and $U_{\overline{S}}:=U\cap\overline{S}$.
Define
\begin{align*}
\y_{U_S}(u) =
\begin{cases}
\frac{d(u)}{\mu(U)}       & \quad \text{if } u\in U_S,\\
-\frac{d(u)}{\mu(U)}  & \quad \text{if } u\in U_{\overline{S}},\\
0 & \quad \text{otherwise}.
\end{cases}
\qquad
\z_{U_S}(u) =
\begin{cases}
\sqrt{\frac{d(u)}{\mu(U)}}       & \quad \text{if } u\in U_S,\\
-\sqrt{\frac{d(u)}{\mu(U)}}  & \quad \text{if } u\in U_{\overline{S}},\\
0 & \quad \text{otherwise}.
\end{cases}
\end{align*}
In particular, we let $\y_S:=\y_{V_S}$ and $\z_S:=\z_{V_S}$.
We will show that for any subgraph $U\subseteq V$ with volume $\mu(U)\geq \left(1-\frac{1}{8}\right)\mu_G$, it holds that
\begin{align}
\y_{U_S}M^{t}{(\1_{U_S}-\1_{U_{\overline{S}}})}^{\top}\geq \frac{1}{60} {(1-2\varepsilon)}^t,
\label{eqn:goodset}
\end{align}

If we have Inequality~\eqref{eqn:goodset}, then
\begin{align*}
  \sum_{v\in U}\frac{d_v}{\mu(U)}\sum_{u\in U}|\1_v M^t(u)|
  & \geq\sum_{v\in U_S}\frac{d_v}{\mu(U)}\1_v M^t{(\1_{U_S} - \1_{U_{\overline{S}}})}^{\top}-\sum_{v\in U_{\overline{S}}}\frac{d_v}{\mu(U)}\1_v M^t {(\1_{U_S} - \1_{U_{\overline{S}}})}^{\top}\\
  &=\y_{U_S}M^t{(\1_{U_S} - \1_{U_{\overline{S}}})}^{\top}\geq \frac{1}{60} {(1-2\varepsilon)}^{t}.
\end{align*}

Therefore, there must exist some vertex $v\in U$ satisfying that $\sum_{u\in U}|\1_v M^t(u)|\geq \frac{1}{60} {(1-2\varepsilon)}^t$, which directly gives that
$$\sum_{u\in V}|\1_v M^t(u)|\geq \frac{1}{60} {(1-2\varepsilon)}^t.$$

Finally, by the choice of $U$, we know that the volume of the set of vertices satisfying~\eqref{eqn:l1lower} is at least $\mu_G/8$.
To see this, because we argue for \emph{any} set $U$ with $\mu(U)\geq (1-\frac{1}{8})\mu_G$, there is a ``good'' vertex. Thus, we can repeatedly find such ``good'' vertices starting from the original graph, until we reach a set of volume less than $(1-\frac{1}{8})\mu_G$.
The set of ``good'' vertices will have volume at least $\frac{1}{8}\mu_G$.
This completes the proof.

Now we prove~\eqref{eqn:goodset}.
First, we note that since $e(S,\overline{S})\geq (1-\varepsilon) m$, we have $e(S)+e(\overline{S})=m-e(S,\overline{S})\leq \varepsilon m$, and thus
\begin{align*}
2 \varepsilon & \geq\frac{4e(S)+4e(\overline{S})}{\mu(V)}=\frac{\sum_{(u,v)\in E}{(\1_{S,\overline{S}}(u)+\1_{S,\overline{S}}(v))}^2}{\sum_u{\1_{S,\overline{S}}(u)}^2d(u)}
=2-\frac{\sum_{(u,v)\in E}{(\1_{S,\overline{S}}(u)-\1_{S,\overline{S}}(v))}^2}{\sum_u{\1_{S,\overline{S}}(u)}^2d(u)}=2-\z_{S}\LL\z_S^T,
\end{align*}
where $\z_S(u):=\sqrt{\frac{d(u)}{\mu_G}}$ if $u\in S$ and $-\sqrt{\frac{d(u)}{\mu_G}}$ if $u\notin S$. If we write $\z_S=\sum_i\alpha_i\vv_i$, then the above gives us that
\begin{align}
\sum_{i=1}^n\lambda_i\alpha_i^2=\z_{S}\LL\z_S^T\geq 2(1-\varepsilon). \label{eqn:eigenvalues}
\end{align}

Now we let $H:=\{i:\lambda_i\geq2{(1-\varepsilon)}^2, 1\leq i\leq n\}$.
Then by the fact that $0=\lambda_1\leq\cdots\leq\lambda_n\leq 2$, $\sum_{i\in H}\alpha_i^2 +\sum_{i\notin H}\alpha_i^2=\norm{\z_S}_2^2=1$, and~\eqref{eqn:eigenvalues}, we have
\[
  2(1-\varepsilon)\leq 2\sum_{i\in H}\alpha_i^2+2{(1-\varepsilon)}^2\sum_{i\notin H}\alpha_i^2=2\sum_{i\in H}\alpha_i^2+2{(1-\varepsilon)}^2\left(1-\sum_{i\in H}\alpha_i^2\right),
\]
which gives that
\[
  \sum_{i\in H}\alpha_i^2\geq \frac{\varepsilon(1-\varepsilon)}{1-{(1-\varepsilon)}^2}=\frac{\varepsilon(1-\varepsilon)}{2\varepsilon - \varepsilon^2}\geq \frac{1}{3},
\]
where the last inequality follows from the fact that $\varepsilon\leq 1/2$.

Now we write $\z_{U_S}=\sum_i\beta_i\vv_i$. We have that
\begin{align*}
  \sum_{i\in H}{(\alpha_i-\beta_i)}^2 & \leq\sum_{i=1}^n{(\alpha_i-\beta_i)}^2=\norm{\z_S-\z_{U_S}}_2^2
  =\sum_{u\in U}{\left(\sqrt\frac{d(u)}{\mu_G}-\sqrt\frac{d(u)}{\mu(U)}\right)}^2+\sum_{u\in\overline{U}}\frac{d(u)}{\mu_G} \\
  &=2-2\sqrt{\frac{\mu(U)}{\mu_G}}
  \leq 2-2\sqrt{1-\frac{1}{8}} \qquad
  \leq 2-2\left(1-\frac{1}{12}\right) =\frac{1}{6},
\end{align*}
where the second to last inequality follows from our assumption that $\mu(U)\geq \left(1- 1/8\right)\mu_G$.

Thus,
\begin{align*}
  \sum_{i\in H}\beta_i^2 \geq {\left(\sqrt{\sum_{i\in H}\alpha_i^2}-\sqrt{\sum_{i\in H}{(\alpha_i-\beta_i)}^2}\right)}^2
  \geq{\left(\sqrt{\frac{1}{3}}-\sqrt{\frac{1}{6}}\right)}^2>\frac{1}{60}.
\end{align*}

Therefore,
\begin{align*}
\y_{U_S}M^{t}{(\1_{U_S}-\1_{U_{\overline{S}}})}^{\top}
&=\frac{1}{2^{t}}\y_{U_S}D^{-1/2}\LL^{t} D^{1/2} {(\1_{U_S}-\1_{U_{\overline{S}}})}^{\top}
=\frac{1}{2^t}\z_{U_S}\LL^{t}\z_{U_S}^T\\
&=\frac{1}{2^t}\sum_i\lambda_i^{t}\beta_i^2
\geq\frac{{(2{(1-\varepsilon)}^2)}^{t}}{2^{t}}\sum_{i\in H}\beta_i^2\geq\frac{1}{60} {(1-2\varepsilon)}^t,
\end{align*}
which completes the proof.
\end{proof}
\subsubsection{Putting things together: Proof of Theorem~\ref{thm:max-cut}}
By Lemma~\ref{lemma:expander}, we have that
$$\norm{\p_v^\ell \cdot D^{-1/2}}_2^2\leq \frac{1}{\mu_G}+{\left(1-\frac{\phi^2}{4}\right)}^{2\ell}\leq \frac{2}{\mu_G},$$
which, together with the fact that $\D_{v,e}=2\p_{v,e}^\ell,  \D_{v,o}=2\p_{v,o}^\ell$ and that $\p_{v,o}^\ell +\p_{v,e}^\ell=\p_v^\ell$, directly implies that
$$\norm{\D_{v,e} \cdot D^{-1/2}}_2^2, \norm{\D_{v,o} \cdot D^{-1/2}}_2^2\leq O\left(\frac{1}{\mu_G}\right).$$

Now set $\xi=\xi_{\mathrm{trm}}, b=\Theta(1/\mu_G),\delta=1/n^2$ in  Theorem~\ref{thm:distribution}. We have the following
\begin{description}
\item[(Completeness)] Consider the case that $\MC(G)\geq 1-\varepsilon$. Then since we sample $O(1)$ vertices and a vertex from $V_g$ is sampled with probability at least $1/8$, where $V_g$ is as defined in Lemma~\ref{lem:maxcut}, with probability at least $5/6$, at least one vertex, say $v$, from the sample set belongs to $V_g$. Then by Lemma~\ref{lemma:completeness}, and Theorem~\ref{thm:distribution}, with probability at least $5/6$, the \textsc{$\ell_2$-DifferenceTest} will reject the pair of distributions $(\D_{v,e},\D_{v,o})$ corresponding to $v$, in which case our algorithm will accept the graph. Thus, with probability at least $2/3$, our algorithm accepts the graph.

\item[(Soundness)] Consider the case that $\MC(G)\leq 1-\rho$ with $\rho\geq \frac{\varepsilon}{C \phi^2}$. By Lemma~\ref{lemma:soundness} and Theorem~\ref{thm:distribution}, with probability at least $2/3$, for any starting vertex $v$, the \textsc{$\ell_2$-DifferenceTest} will accept the pair of distributions $(\D_{v,e},\D_{v,o})$ corresponding to $v$, in which case our algorithm will reject the graph.
\end{description}
Finally, note that the query complexity and running time are dominated by $O(1)$ invocations of \Call{$\ell_2$-DifferenceTest}{}.
By our setting of parameters, the total time and query complexity are $\tilde{O}(\mu_G^{1/2+\varepsilon/(2C \phi^2\rho)}/(\phi^2\rho))$.


\section{\texorpdfstring{Max E2Lin$(q)$}{Max E2Lin(q)} on Expanders}\label{sec:max-2lin}
In this section, we design sublinear-time algorithms for \textsc{Max E2LIN$(q)$} and prove Theorem~\ref{thm:max-2lin}.

\subsection{The Algorithm}

Our algorithm is based on a result of testing $k$-clusterability. The following result was implicit in~\cite{CKKMP18:cluster}, and we give a proof sketch here.
\begin{theorem}[\cite{CKKMP18:cluster,CKKMP18:cluster_arxiv}]\label{thm:ckkmp}
	Let $\varepsilon\leq c_1\lambda$ for some constant $c_1\in(0,1)$. Let $G$ be a graph with $m$ edges. There exists an algorithm \Call{TestClusterability}{$G,k,\lambda, \varepsilon$} that with probability at least $2/3$,
	\begin{description}
    \itemsep=0pt
		\item[(Completeness)] accepts if $\lambda_{k+1}(G) \geq \lambda$,
		\item[(Soundness)] rejects if $G$ contains $k+1$ pairwise disjoint subsets $C_1,\ldots, C_{k+1}$, each of volume $\frac{\mu_G}{k+1}$ and conductance at most $\varepsilon$.
	\end{description}
	The query complexity of the algorithm is $O(\poly(k\log m/\lambda)\cdot m^{1/2+O(\varepsilon/\lambda)})$.
\end{theorem}
\begin{proof}[Proof Sketch]
  The proof directly follows from the proof of Theorem~1 in~\cite{CKKMP18:cluster_arxiv}, which gives an algorithm for distinguishing $(k,\varphi_{\mathrm{in}})$-clusterable graphs from $(k,\varphi_{\mathrm{out}},\beta)$-unclusterable graphs, where a graph $G$ is said to be $(k,\varphi_{\mathrm{in}})$-clusterable, if $G$ admits an $h$-partition $C_1,\ldots, C_h$ for $h\leq k$, such that the so-called \emph{internal conductance}\footnote{The internal conductance of a set $C$ in $G=(V,E)$ is defined to be the conductance of the graph $G'[C]$ that is formed by adding an appropriate number of self-loops to each vertex in the induced subgraph $G[C_i]$ so that each vertex $v$ in $G'[C]$ has degree $d_G(v)$.} of each $C_i$ is at least $\varphi_{\mathrm{in}}$; a graph $G$ is said to be $(k,\varphi_{\mathrm{out}},\beta)$-unclusterable if $G$ contains $k+1$ pairwise disjoint subsets $C_1,\ldots,C_{k+1}$ such that for each $i\leq k+1$, $\mu_G(C_i)\geq \beta\cdot \frac{\mu_G}{k+1}$ and $\phi_G(C_i)\leq \varphi_{\mathrm{out}}$. More precisely, Theorem~1 in~\cite{CKKMP18:cluster_arxiv} says if $\varphi_{\mathrm{out}}\leq \frac{1}{480}\varphi_{\mathrm{in}}^2$, then there exists an algorithm \Call{TestClusterability}{$G,k,\varphi_{\mathrm{in}},\varphi_{\mathrm{out}},\beta$} that distinguishes a $(k,\varphi_{\mathrm{in}})$-clusterable graph from a $(k,\varphi_{\mathrm{out}},\beta)$-unclusterable graph, with success probability at least $\frac23$ and makes $O(\poly(k\log m/\beta)\cdot m^{1/2+O(\varphi_{\mathrm{out}}/\varphi_{\mathrm{in}}^2)})$ queries.

We note that in~\cite{CKKMP18:cluster_arxiv}, the proof for their Theorem~1 holds as long as $\lambda_{k+1}(G)\geq \varphi_{\mathrm{in}}^2/2$ (which was guaranteed by their combinatorial condition and proven in Lemma~10 in~\cite{CKKMP18:cluster_arxiv}).
  That is, if we replace $\varphi_{\mathrm{in}}, \varphi_{\mathrm{out}},\beta$ by $\sqrt{2\lambda},\varepsilon,1$, respectively, in Theorem~1 in~\cite{CKKMP18:cluster_arxiv}, then the statement of the theorem follows.
\end{proof}

Our algorithm for \textsc{Max E2Lin$(q)$}, given in Algorithm~\ref{alg:max-2lin}, simply invokes \Call{TestClusterability}{} on the \emph{label-extended graph} $G_\mathcal{I} = (V_\mathcal{I},E_{\mathcal{I}})$ of the instance $\II=(G=(V,E),q,\bc)$, which is defined as
\[
  V_\mathcal{I} = V \times \mathbb{Z}_q
  \quad\text{and}\quad
  E_\mathcal{I} = \{\{(u,i),(v,i+\bc_{uv} )\} \mid (u,v) \in E, i \in \mathbb{Z}_q\}.
\]

\begin{algorithm}[t!]
	\caption{\textsc{TestExpanderMLin}$(G,q,\bc,\phi,\varepsilon,\rho)$}\label{alg:max-2lin}
	Let $\lambda:= \Omega(\rho\phi^2/q^{12})$\;
	Invoke the algorithm \Call{TestClusterability}{$G_\II, q-1,\lambda,\varepsilon/2$} from~\cite{CKKMP18:cluster}\;
	\lIf{$G$ is rejected by \Call{TestClusterability}{} }{Output \textbf{Accept}}
	\lElse{
	Output \textbf{Reject}}
\end{algorithm}

\paragraph{Implementation of the algorithm} To invoke \Call{TestClusterability}{} on the label-extended graph, we only need to sample a vertex $(v,i)$ from $G_\II$ with probability proportional to the degree of $(v,i)$, and perform lazy random walks on $G_\II$. The former can be done by sampling a vertex $v$ with probability proportional to $d(v)$ in $G$ (which in turn can be done in the same way as we did in \Call{TestExpanderMC}{}) and then randomly sampling a label $i\in \mathbb{Z}_q$. For the latter, at each step of the random walk, we stay at the current vertex $(v,i)$ with probability $1/2$, and with the rest half probability, we randomly sample a neighbor $u$ of $v$ in $G$, and then jump to the corresponding neighbor $(u,i+\bc_{uv})$ in $G_\II$.

\subsection{Proof of Theorem~\ref{thm:max-2lin}}

\subsubsection{Completeness}
\begin{lemma}\label{lem:lincompleteness}
  Let $\mathcal{I}=(G,q,\bc)$ be an instance of \textsc{Max E2Lin$(q)$} with $\mathsf{OPT}(\mathcal{I}) > 1-\varepsilon$.
  Then, $V(G_\mathcal{I})$ can be partitioned into $q$ vertex sets, each of volume $\mu_{G_\mathcal{I}}/q$ and conductance at most $\varepsilon/2$.
\end{lemma}
\begin{proof}
  Let $\psi \colon V \to \mathbb{Z}_q$ be an optimal assignment of $\II$ and thus $\psi$ satisfies at least $1-\varepsilon$ fraction of the constraints. 
  Then, define $S_0 = \{(v,\psi(v)) \mid v \in V\}$ and $S_i = \{(v,j + i) \mid (v,j) \in S_0 \}$ for $i =1,\ldots,q-1$.
  Clearly we have $\mu_G(S_i) = \mu_{G_\mathcal{I}}/q$ for every $i \in \{0,1,\ldots,q-1\}$. Note that for each $i$, $e_{G_{\II}}(S_i,\bar{S_i})<\varepsilon$ by the definition of $\psi$ and $S_i$. 
  Now for every $i \in \{0,1,\ldots,q-1\}$, we have
  \[
    \phi_{G_\mathcal{I}}(S_i) = \phi_{G_\mathcal{I}}(S_0) < \frac{\varepsilon m}{2m} = \frac{\varepsilon}{2},
  \]
  where $m$ is the number of edges in $G$.
\end{proof}

\subsubsection{Soundness}

\begin{lemma}\label{lemma:linsoundness_combinatorial}
  Let $\mathcal{I}=(G=(V,E),q,\bc)$ be an instance of \textsc{Max E2Lin$(q)$} with $\phi_G \geq \phi$ and $\mathsf{OPT}(\mathcal{I}) < 1-\rho$.
  Let $G_\II=(V' = V \times \mathbb{Z}_q,E')$ be the label-extended graph of $\mathcal{I}$.
  Then, $\phi_{G_\II}(q) \geq \rho \phi/6q$.
\end{lemma}
\begin{proof}
  Let $S' \subseteq V'$ be a vertex set with $\mu_{G_\II}(S') \leq \mu_{G_\II} / q$.
  Suppose $\phi_{G_\II}(S') < \varepsilon \phi/q$, where $\varepsilon$ will be determined later.
  Define
  \begin{align*}
   A' & = \{(v,i) \in S' \mid \text{there exists no } j \neq i\text{ s.t.\ } (v,j) \in S'\}, \\
   B' & = S' \setminus A',\\
   A & = \{v \in V \mid \text{there exists }i \in \mathbb{Z}_q \text{ s.t.\ } (v,i) \in A'\},\\
   B & = \{v \in V \mid \text{there exists }i \in \mathbb{Z}_q \text{ s.t.\ } (v,i) \in B'\}, \\
   S & = A \cup B.
  \end{align*}
  \begin{claim}
    Let $\alpha \in [0,1]$ be such that $\alpha \cdot \mu_G(A) = \mu_G(S)$.
    Then, $\alpha > 1-\varepsilon$.
  \end{claim}
  \begin{proof}
    Let $e$ be an edge leaving $B$ in $G$.
    Then, there exist at least two corresponding edges, say, $e'_1$ and $e'_2$, leaving $B'$ in $G_\II$.
    Because $A'$ can incident to at most one of them, at least one of them must leave $S' = A' \cup B'$.
    Hence, we have
    \begin{align*}
      e_{G_\II}(S',V' \setminus S') \geq \phi \cdot \min\{\mu_G(B), \mu_G(V) - \mu_G(B)\}.
    \end{align*}

    We note that
    \begin{align*}
      \mu_{G_\II}(S') & \leq \frac{\mu_{G_\II}(V')}{q} = \mu_G(V),\\
      \mu_{G_\II}(S') & \geq \mu_G(A) + 2\mu_G(B) \geq \alpha \cdot \mu_G(S) + 2(1-\alpha) \cdot \mu_G(S) = (2-\alpha) \cdot \mu_G(S),
    \end{align*}
    and hence $\mu_G(V)/\mu_G(S) \geq 2-\alpha$.

    Then, since $\mu_G(B)=(1-\alpha)\mu_G(S)$, we have
    \begin{align*}
      \phi_{G_\II}(S') & = \frac{e_{G_\II}(S',V' \setminus S')}{\mu_{G_\II}(S')} \geq \frac{\phi \cdot \min\{(1-\alpha) \cdot \mu_G(S), \mu_G(V)-(1-\alpha) \cdot \mu_G(S)\}}{q \cdot \mu_G(S)} \\
      & = \frac{\phi}{q} \cdot \min\left\{1-\alpha, \frac{\mu_G(V)}{\mu_G(S)}-1+\alpha \right\}
      \geq \frac{\phi}{q} \cdot \min\left\{1-\alpha, 1 \right\} = \frac{\phi}{q} \cdot (1-\alpha).
    \end{align*}
    Because $\phi_{G_\II}(S') < \epsilon \phi/q$, we have $\alpha > 1-\varepsilon$.
  \end{proof}

  \begin{claim}
    Let $\sigma \in [0,1]$ be such that $\sigma \cdot \mu_G(V) = \mu_G(S)$.
    Then, we have $\sigma > 1-\varepsilon$.
  \end{claim}
  \begin{proof}
    We have
    \begin{align*}
      e_{G_\II}(S',V' \setminus S') \geq \phi \min\{\mu_G(S), \mu_G(V)-\mu_G(S)\} = \phi \min\{\sigma,1-\sigma\} \cdot \mu_G(V).
    \end{align*}
    Then,
    \begin{align*}
      \phi_{G_\II}(S') & =\frac{e_{G_\II}(S',V' \setminus S')}{\mu_{G_\II}(S')}
      \geq \frac{e_{G_\II}(S',V' \setminus S')}{q \cdot \mu_G(S)}
      \geq \frac{\phi \min\{\sigma,1-\sigma\} \cdot \mu_G(V)}{q \sigma \cdot \mu_G(V)}
      = \frac{\phi}{q} \min\left\{1,\frac{1}{\sigma}-1\right\}
    \end{align*}
    It follows that $\varepsilon > \min\{1,1/\sigma-1\}$ and hence $\sigma > 1/(1+\varepsilon) \geq 1-\varepsilon$.
  \end{proof}

  Let $f:V \to \mathbb{Z}_q \cup \{\bot\}$ be the partial assignment induced by $A$.
  Note that
  \begin{align*}
    \mu_G(f^{-1}(\bot)) & = \mu_G(V \setminus A)
    = \mu_G(V\setminus S) + \mu_G(B)
    = (1-\sigma)\mu_G(V) + (1-\alpha)\mu_G(S) \\
    & \leq (2-\sigma-\alpha)\mu_G(V)
    \leq 2\varepsilon \cdot \mu_G(V).
  \end{align*}

  The fraction of unsatisfied constraints is at most
  \begin{align*}
    \frac{\mu_G(f^{-1}(\bot))}{\mu_G(V)/2} + \frac{e_{G_\II}(A' , (A \times \mathbb{Z}_q) \setminus A')}{\mu_G(V)/2}
    \leq \frac{2\mu_G(f^{-1}(\bot))}{\mu_G(V)} + \frac{2e_{G_\II}(S' , V' \setminus S')}{\mu_{G_\II}(V')/q}
    \leq 4\varepsilon + 2\varepsilon \phi = 2\varepsilon(2+\phi) \leq 6\varepsilon.
  \end{align*}
  Hence, the claim holds by setting $\varepsilon = \rho/6$.
\end{proof}

\begin{lemma}\label{lemma:linsoundness_eigen}
	Let $G_\II$ be defined as in Lemma~\ref{lemma:linsoundness_combinatorial}. Let $\lambda_i$ be the $i$-th smallest eigenvalue of $\LL_{G_\II}$. Then $\lambda_q\geq \Omega(\frac{\rho^2 \phi^2}{q^6})$.
\end{lemma}
\begin{proof}
  By Lemma~\ref{lemma:linsoundness_combinatorial} and applying Theorem~\ref{thm:higher-order-cheeger} with $k=q$, we have
  \[
    \lambda_{q}\geq \Omega\left(\frac{{\rho_{q}(G_\II)}^2}{{q}^4}\right) \geq \Omega\left(\frac{{\phi_{G_\II}(q)}^2}{q^4}\right) =\Omega\left(\frac{\rho^2 \phi^2}{q^6} \right). \qedhere
  \]
\end{proof}

\subsection{Putting things together: Proof of Theorem~\ref{thm:max-2lin}}

\begin{proof}[Proof of Theorem~\ref{thm:max-2lin}]
In the soundness, by Lemma~\ref{lemma:linsoundness_eigen}, $\lambda_q(G_\II)\geq \lambda:= \Omega(\rho^2 \phi^2/q^6)$.

In the completeness, by Lemma~\ref{lem:lincompleteness}, $G_\II$ contains $q$ pairwise disjoint subsets, each of volume $\mu_G/q$ and with conductance at most $\varepsilon$.

Now since $\rho=\Omega(q^3\sqrt{\varepsilon}/\phi)$, we can guarantee that $\varepsilon \leq c_1\lambda$. Then the statement of the theorem follows by applying Theorem~\ref{thm:ckkmp} with $k=q-1$ on graph $G_\II$.
\end{proof}


\section{Unique Label Cover on Expanders}\label{sec:unique-label-cover}
In this section, we design sublinear-time algorithm for \textsc{Unique Label Cover} and prove Theorem~\ref{thm:unique-game}.

\subsection{Preliminaries}
Our algorithm underlying Theorem~\ref{thm:unique-game} makes use of label-extended graphs of instances of \textsc{Unique Label Cover} and a sublinear-time clustering oracle, which are introduced below.
\paragraph{Label-Extended Graphs}

Let $\mathcal{I}=(G=(V,E),q,\pi)$ be an instance of \textsc{Unique Label Cover}.
Then, we define its \emph{label-extended graph} $G_\mathcal{I}$ as
\[
V(G_\mathcal{I}) = V \times [q]
\quad\text{and}\quad
E(G_\mathcal{I}) = \{((u,i),(v,\pi(i))) \mid (u,v) \in E, i \in [q]\}.
\]
This matches the one in Section~\ref{sec:max-2lin} for \textsc{Max E2Lin$(q)$}.
In the following, we use $G'$ to denote the label-extended graph $G_\II$, and use $V',E'$ to denote the vertex set $V(G_\II)$ and the edge set $E(G_\II)$, respectively.

\paragraph{Sublinear-Time Clustering Oracle} In order to describe the algorithm, we introduce a sublinear-time clustering oracle given by Gluch et al.~\cite{gluch2021spectral}.

\begin{definition}\label{def:clusterablegraphs}
	Given positive integers $d, k$ and $\alpha,\beta \in [0,1]$, we call a $k$-partition $C_1,\ldots,C_k$ of a $d$-bounded graph $G$ a \emph{$(k,\alpha,\beta)$-clustering} if for each $i \in \{1,\ldots,k\}$, $\phi_{G[C_i]}\geq \alpha$ and $\phi_G(C_i) \leq \beta$.
	A $d$-bounded graph $G$ is called to be \emph{$(k,\alpha,\beta,q)$-clusterable} if $G$ has an $(k,\alpha,\beta)$-clustering $C_1,\dots,C_k$ such that $\min_i \mu_G(C_i)\geq \frac{\mu_G}{q}$.
\end{definition}
Note that when we talk about a $(k,\alpha,\beta,q)$-clusterable graph, it always holds that $k\leq q$, as the volume of the minimum cluster is at least $\frac{\mu_G}{q}$, which implies that there are at most $q$ clusters. We will make use of the following result.
\begin{theorem}[\cite{gluch2021spectral}]\label{thm:clusteringoracle}
Let $d\geq 3$ be a constant, $2\leq k\leq q$, $\alpha \in (0,1)$, and $\beta\ll\frac{\alpha^3}{q^{10}}$.
Let $G$ be an $n$-vertex and $d$-bounded graph that is $(k,\alpha,\beta,q)$-clusterable, then there exists a clustering oracle $\mathcal{O}$ for $G$ that
	\begin{itemize}
		\item has $\widetilde{O}((d/\alpha )^{O(1)}\cdot 2^{O((\alpha^2/\beta) \cdot q^{100})}\cdot n^{1/2+O(\beta/\alpha^2)})$ preprocessing time,
		\item has $\widetilde{O}({(qd/\alpha\beta)}^{O(1)}\cdot n^{1/2+O(\beta/\alpha^2)})$ query time,
		\item and for a $(k,\alpha,\beta,q)$-clustering $C_1,\ldots,C_k$, the oracle $\mathcal{O}$ provides consistent query access to a partition\footnote{That is, for any queried vertex $v$, the oracle $\mathcal{O}$ returns the index $i$ with $v\in \widehat{C}_i$.} $(\widehat{C}_1,\ldots,\widehat{C}_k)$ of $V$, such that with probability at least $\frac{1}{n^2}$ (over the random bits of $\mathcal{O}$),
		it holds that for some permutation $\tau:[k]\to [k]$, for any $i\in[k]$,
		\[\mu_{G}(C_i\triangle \widehat{C}_{\tau(i)})\leq O_d\left(\frac{\beta \cdot q^{10}}{\alpha^3}\right)\mu_G(C_i).\]
	\end{itemize}
\end{theorem}

We remark that the above theorem was not explicit in~\cite{gluch2021spectral}, as the sizes of underlying clusters in our setting may differ by a factor of at most $q$, in comparison to $\Theta(1)$ as assumed in~\cite{gluch2021spectral} and the conductance of a set defined in~\cite{gluch2021spectral} differs from ours by a factor of $d$. However, it is easy to modify their argument to prove Theorem~\ref{thm:clusteringoracle}, and we sketch the main differences and the modifications we need from~\cite{gluch2021spectral} in Section~\ref{sec:proofclusteringoracle}.
Furthermore, we remark that the main result in~\cite{gluch2021spectral} is stated with a tradeoff between preprocessing time and query time, while we only stated the special case that these two times are of the same asymptotic order for simplicity.

We need a subroutine \Call{TestOuterConductance}{$G,n,d,q,\alpha,\beta, \mathcal{O}, i$} to test if a potential cluster $\widehat{C}_i$ has a small outer conductance. The main idea of this subroutine is as follows. We sample vertices uniformly at random and check whether they are contained in the cluster $\widehat{C}_i$ using the spectral clustering oracle $\mathcal{O}$. If so, we sample a random edge incident to the sample vertex, then we obtain a random edge incident to a random vertex from this cluster. Then we can use the fraction of edges leaving $\widehat{C}_{i}$ to estimate the outer conductance of $\widehat{C}_i$. The pseudocode of this algorithm is given in Algorithm~\ref{alg:testouterconductance}.

\begin{algorithm}[t!]
	\caption{\textsc{TestOuterConductance}$(G,n,d,q,\alpha,\beta, \mathcal{O}, i)$}\label{alg:testouterconductance}
	Let $a:=0$, $b:=0$\; 
	\ForEach{$t=1,\dots,s := \Theta(\frac{\alpha^2 \cdot qd\log (n)}{\beta})$}{
		Sample a vertex $x$ from $V(G)$\;
		With probability $\frac{\deg_G(x)}{d}$, sample a neighbor $y$ of $x$; otherwise, let $y=x$\;
		\If{$\mathcal{O}(G,x)=i$}{
			$b\gets b+1$\;
			\If{$\mathcal{O}(G,y)\neq i$}{
				$a\gets a+1$\;}	
		}
	}
	\Return $\frac{a}{b}$\;
\end{algorithm}

The following can be derived from Lemma 44 in~\cite{gluch2021spectral}.

\begin{lemma}[\cite{gluch2021spectral}]\label{lem:testouterconductance}
	Let $G$ be an $n$-vertex and $d$-bounded graph.
	Let $\widehat{C}_1,\dots,\widehat{C}_k$ be the clusters (implicitly)
	output by the spectral clustering oracle $\mathcal{O}$ when given $G$. Suppose that $|\widehat{C}_i|\geq \frac{n}{10dq}$ for some $q\ge k$.
	Then, \Call{TestOuterConductance}{$G,n,d,q,\alpha,\beta, \mathcal{O}, i$} outputs an estimate $\eta$ such that
	\[
	\eta \in \left[\frac12 \frac{\phi_G(\widehat{C}_i)}{d}-\frac{\beta}{\alpha^2}, \frac32 \phi_G(\widehat{C}_i)+\frac{\beta}{\alpha^2}\right].
	\]
	with probability $1-\frac{1}{n^{O(1)}}$.
	The running time of \Call{TestOuterConductance}{$G,n,d,q,\alpha,\beta, \mathcal{O}, i$}, including all the invocations of $\mathcal{O}$, is $\widetilde{O}({(\frac{dq}{\alpha\beta})}^{O(1)}\cdot n^{1/2+O(\beta/\alpha^2)})$.
\end{lemma}
The above lemma follows from the proof of Lemma 44 of~\cite{gluch2021spectral}. Here we give the proof for the sake of completeness.
\begin{proof}[Proof of Lemma~\ref{lem:testouterconductance}]
Firstly, we note that since $|\widehat{C}_i|\geq \frac{n}{10qd}$, by the Chernoff bound, it holds that with  probability $1-\frac{1}{n^{O(1)}}$, the number of sampled vertices $x$ that belong to $\widehat{C}_i$ is $b\geq \frac{s}{20qd}$. Note that conditioned on the event that the sampled vertex $x$ belongs to $\widehat{C}_i$, $x$ is uniformly distributed on the set $\widehat{C}_i$,
and the probability that the sampled neighbor $y$ belongs to $V\setminus \widehat{C}_i$ is $\psi(\widehat{C}_i):=\frac{e_{G'}(\widehat{C}_i,V\setminus \widehat{C}_i)}{d\cdot |\widehat{C}_i|}$.

Conditioned on the event that $b\geq \frac{s}{20qd}$, by the Chernoff-Hoeffding bound, with probability at least $1-2e^{-\Omega(b\beta/\alpha^2)}\geq 1-\frac{1}{n^{O(1)}}$, it holds that
\[
\frac{a}{b} \in \left[\frac12 \psi(\widehat{C}_i)-\frac{\beta}{\alpha^2}, \frac32 \psi(\widehat{C}_i)+\frac{\beta}{\alpha^2}\right]
\]
Then the correctness of the lemma follows from the fact that $\frac{\phi_G(\widehat{C}_i)}{d}\leq \psi(\widehat{C}_i)\leq \phi_G(\widehat{C}_i)$. Finally, we note that the running time of the algorithm is dominated by $O(s)$ invocations of the oracle $\mathcal{O}$, and is thus $\widetilde{O}({(\frac{dq}{\alpha\beta})}^{O(1)}\cdot n^{1/2+O(\beta/\alpha^2)})$ by Theorem~\ref{thm:clusteringoracle}.
\end{proof}

\subsection{Algorithm Description}
Now we describe our algorithm for solving \textsc{Unique Label Cover} on expander graphs. The algorithm starts with estimating the total volume of the underlying graph $G$ by sampling. Then the algorithm samples some vertices, and invokes \textsc{SpectralClusteringOracle} from~\cite{gluch2021spectral} to obtain the clustering membership information of each sampled vertices, and then tests if there exists a subset $S$ of volume $\frac{\mu_{G'}}{q}=\mu_G$ with outer conductance at most $O(\varepsilon)$. In the completeness, we can guarantee that such a subset $S$ always exists while in the soundness, there is no subset of volume $\frac{\mu_{G'}}{q}$ of small outer conductance.
The pseudocode of the algorithm is described in \cref{alg:uniquegames}.

\begin{algorithm}[t!]
	\caption{\textsc{UniqueLabelCover}$(G,q,\pi,n,d,\phi,\varepsilon,\rho)$}\label{alg:uniquegames}
For any $r$ with $2\leq r\leq q+1$, define
	$f(r):=\left(\frac{\varepsilon}{2q^{20}}\right)^{4^{2-r}}\cdot q^{100-40r}\cdot \phi^{\frac{r-2}{q-1}}$\;
Sample a set $T\subseteq V$ of $\Theta(\frac{dq \log n}{\xi_0^2})$ vertices, where $\xi_0:=O(\frac{ q^{50}\cdot \varepsilon^{4^{1-r}}}{\phi^{\frac{2r-1}{q-1}}})$\;
Define $x:=\frac{n}{|T|}\sum_{v\in T}\deg_G(v)$ to be the estimate of $\mu_{G}$\label{alg:estimateofvolumeGprime} \;
	\ForEach{$r=2,\dots,q$}{
 Let $\alpha:= \frac{f(r+1)}{30r}, \beta:=r\cdot f(r)$,
and $\xi:=O_d(\frac{\beta q^{10}}{\alpha^3})$\;
	Sample $s:=\Theta(dq \log n)$ vertices from $V$\;
	For each sampled vertex $v$, make a query to $\mathcal{O}$ from Theorem~\ref{thm:clusteringoracle} to determine which set $\widehat{C}_i$ it belongs to\;
	For each $i\in [r]$, let $f_i$ be the sum of degrees of sampled vertices that are reported to belong to $\widehat{C}_i$\;
	Let $s_i:=\frac{nq}{s}\cdot f_i$ be the estimate of the volume  $\mu_{G'}(\widehat{C}_i)$\label{alg:definefisi}\;
	\If{$s_i\geq \frac{x\cdot q}{4(q+1)}$ for all $i \in \{1,\ldots,r\}$}{
		\If{there exists $i$ with $s_i\in [(1-\xi)x, (1+\xi)x]$ and $\Call{TestOuterConductance}{G',nq,d,q,\alpha,\beta, \mathcal{O}, i} \leq O_d(\frac{q^{85q}\cdot \varepsilon^{4^{1-r}}}{\phi^{\frac{2r-1}{q-1}}})$}{Output \textbf{Accept}}}
}
	{
		Output \textbf{Reject}}
\end{algorithm}

\subsection{Analysis of \cref{alg:uniquegames}: Technical Lemmas}

Now, we analyze the correctness, query complexity, running time of \cref{alg:uniquegames}.
We start with a property of the label-extended graph of an instance of \textsc{Unique Label Cover} when the underlying graph is an expander.

\begin{lemma}[Expanding property of label-extended graphs]\label{lemma:ugphiq1}
	Let $\mathcal{I}=(G=(V,E),q,\pi)$ be an instance of \textsc{Unique Label Cover} with $\phi_G \geq \phi$.
	Let $G'=(V' = V \times [q],E')$ be the label-extended graph of $\mathcal{I}$.
	Then, $\phi_{G'}(q+1) \geq \frac{\phi}{q(q+1)}$.
\end{lemma}
\begin{proof}
	Consider a set $S'\subseteq V'$ such that $\mu_G(S)\leq \frac{\mu_{G'}(V')}{q+1}$. Consider its projected set $S=\{v \in V \mid \text{there exists }i \in [q] \text{ s.t.\ }(v,i) \in S\}$. Note that $\mu_G(S)\geq \mu_{G'}(S')/q$. Furthermore,
	$\mu_G(S)\leq \frac{q \mu_G(V)}{q+1}$ and hence $\mu_G(V \setminus S)\geq \frac{ \mu_G(V)}{q+1}\geq \frac{ \mu_G(S)}{q+1}$.
	Thus, we have
	\[
	e_G(S,V\setminus S) \geq \phi \min\{\mu_G(S), \mu_G(V \setminus S)\}\geq \frac{\phi\mu_G(S)}{q+1}.
	\]
	Then, we have
	\[
	e_{G'}(S',V'\setminus S') \geq  \frac{\phi\mu_G(S)}{q+1}\geq \frac{\phi\mu_{G'}(S')}{q(q+1)}.
	\]
	Therefore, we have $\phi_{G'}(S')\geq \frac{\phi}{q(q+1)}$.
\end{proof}

\subsubsection{Completeness}
Let $\II$ be an instance of \textsc{Unique Label Cover} with $\mathsf{OPT}(\mathcal{I})\geq 1-\varepsilon$. We first show that there exists a subset $S$ in $G'$ with a small conductance. Then we show that given the gap between $\varepsilon$ and $\phi$, the graph $G'$ admits an $(r,\alpha,\beta)$-clustering for appropriately chosen parameters $r,\alpha,\beta$. Afterwards, we show that exactly one part in the clustering is very close to the subset $S$, which allows us to detect the existence of a large subset of a small outer conductance, by using the spectral clustering oracle.
\begin{lemma}\label{lem:ugcompleteness}
	Let $\mathcal{I}$ be an instance of \textsc{Unique Label Cover} with $\mathsf{OPT}(\mathcal{I}) \geq 1-\varepsilon$.
	Then, there exists a subset $S \subseteq V(G')$ of volume $\mu_{G'}/q$ and conductance at most $\varepsilon/2$.
\end{lemma}
\begin{proof}
	Let $f:V \to [q]$ be an assignment satisfying at least a $(1-\varepsilon)$-fraction of the constraints.
	Consider a set $S = \{(v,f(v)) : v \in V \}$.
	Clearly, the volume of $S$ in $G'$ is $\mu_G=\mu_{G'}/q$.
	Also, we have
	\[
	e_{G'}(S, \bar{S}) \leq \varepsilon \cdot \frac{\mu_G}{2},
	\]
	and hence the conductance of $S$ is
	\[
	\frac{e_{G'}(S, \bar{S})}{\mu_G} \leq \frac{\varepsilon}{2}.
	\qedhere
	\]
	
\end{proof}

Now we show that there exists a clustering of $G'$ such that each part has a large inner conductance and a small outer conductance.
We first state a theorem by Oveis Gharan and Trevisan~\cite{gharan2014partitioning}.
\begin{theorem}[\cite{gharan2014partitioning}]\label{thm:ogtrevisan}
If $\rho_G(k+1)\geq (1+\eta) \rho_G(k)$ for some $\eta \in (0,1)$, then
there exists a $k$-partitioning of $V$ that is a $(k,\eta\cdot  \frac{\rho_G(k+1)}{14k}, k\rho_G(k))$-clustering. 
\end{theorem}

Now we are ready to state the lemma regarding the existence of a good clustering of $G'$. We first introduce a function $f$ for our analysis. For any $r$ with $2\leq r\leq q+1$, we define 
\begin{eqnarray}
f(r):=\left(\frac{\varepsilon}{2q^{20}}\right)^{4^{2-r}}\cdot q^{100-40r}\cdot \phi^{\frac{r-2}{q-1}}\label{def:functionf}
\end{eqnarray}
We have the following facts.
\begin{fact}\label{fact:oneaboutf}
It holds that $f(2)=\frac{\varepsilon}{2}$ and $f(q+1)<\frac{\phi}{q^{10}}$.
\end{fact}
\begin{proof}
By definition $f(2)=\frac{\varepsilon}{2}$. Since $q\geq 2$, we have $f(q+1)<q^{100-120}\phi<\frac{\phi}{q^{10}}$.
\end{proof}
\begin{fact}\label{fact:twoaboutf}
Let $\varepsilon\ll \left(\frac{\phi^2}{q^{100}}\right)^{4^{q-1}}$. Let $c_2>1$ be some sufficiently large constant and let $c_1>0$ be some sufficiently small constant.  Let $2\leq r\leq q$. If we let $\alpha=\frac{f(r+1)}{30r}$ and $\beta=r f(r)$, then 
\begin{enumerate}
\item $\beta<\frac{\phi}{q^{10}}$;
\item $\alpha> 7\varepsilon$;
\item $\frac{\beta}{\alpha^2}\in[\frac{c_1\cdot q^{40}\cdot\varepsilon^{4^{1.5-r}}}{\phi^{\frac{r}{q-1}}},  \frac{c_2\cdot q^{40q}\cdot \varepsilon^{4^{1.5-r}}}{\phi^{\frac{r}{q-1}}}]$;
\item $\frac{\beta}{\alpha^3}\in[\frac{c_1\cdot q^{40}\cdot \varepsilon^{4^{1-r}}}{\phi^{\frac{2r-1}{q-1}}}, \frac{c_2\cdot q^{80q}\cdot \varepsilon^{4^{1-r}}}{\phi^{\frac{2r-1}{q-1}}}]$.
\end{enumerate}
\end{fact}
\begin{proof}
\begin{enumerate}
\item Note that $\beta=rf(r)=r\cdot \left(\frac{\varepsilon}{2q^{20}}\right)^{4^{2-r}}\cdot q^{100-40r}\cdot \phi^{\frac{r-2}{q-1}}\leq\max\{\varepsilon,\frac{\varepsilon^{4^{1-q}}}{q^{10}}\}<\frac{\phi}{q^{10}}$, as $\varepsilon\ll \left(\frac{\phi^2}{q^{100}}\right)^{4^{q-1}}$.
\item Note that $\alpha=\frac{f(r+1)}{30r}=\left(\frac{\varepsilon}{2q^{20}}\right)^{4^{1-r}}\cdot \frac{q^{60-40r}\cdot \phi^{\frac{r-1}{q-1}}}{30q} \geq \left(\frac{\varepsilon}{2q^{20}}\right)^{1/4}\cdot \frac{\phi}{30 q^{21}}>7\varepsilon$, as $\varepsilon\ll \left(\frac{\phi^2}{q^{100}}\right)^{4^{q-1}}$.
\item By definition, 
\[\frac{\beta}{\alpha^2}=\frac{900r^3f(r)}{f(r+1)^2}=\frac{900r^3 \left(\frac{\varepsilon}{2q^{20}}\right)^{4^{2-r}}\cdot q^{100-40r}\cdot \phi^{\frac{r-2}{q-1}}}{\left(\left(\frac{\varepsilon}{2q^{20}}\right)^{4^{1-r}}\cdot q^{60-40r}\cdot \phi^{\frac{r-1}{q-1}}\right)^2}=\frac{900r^3\cdot q^{40r-20}\cdot \left(\frac{\varepsilon}{2q^{20}}\right)^{4^{1.5-r}}}{\phi^{\frac{r}{q-1}}}
\]
Since $2\leq r\leq q$, we have that 
\[
\frac{c_1\cdot q^{40}\cdot\varepsilon^{4^{1.5-r}}}{\phi^{\frac{r}{q-1}}}\leq \frac{\beta}{\alpha^2}\leq \frac{c_2\cdot q^{40q}\cdot \varepsilon^{4^{1.5-r}}}{\phi^{\frac{r}{q-1}}}
\]
\item By definition, 
\[\frac{\beta}{\alpha^3}=\frac{27000r^4f(r)}{f(r+1)^3}=\frac{27000r^4 \left(\frac{\varepsilon}{2q^{20}}\right)^{4^{2-r}}\cdot q^{100-40r}\cdot \phi^{\frac{r-2}{q-1}}}{\left(\left(\frac{\varepsilon}{2q^{20}}\right)^{4^{1-r}}\cdot q^{60-40r}\cdot \phi^{\frac{r-1}{q-1}}\right)^3}=\frac{27000r^4\cdot q^{80r-80}\cdot \left(\frac{\varepsilon}{2q^{20}}\right)^{4^{1-r}}}{\phi^{\frac{2r-1}{q-1}}}
\]
Since $2\leq r\leq q$, we have that 
\[
\frac{c_1\cdot q^{40}\cdot \varepsilon^{4^{1-r}}}{\phi^{\frac{2r-1}{q-1}}}\leq \frac{\beta}{\alpha^3}\leq \frac{c_2\cdot q^{80q}\cdot \varepsilon^{4^{1-r}}}{\phi^{\frac{2r-1}{q-1}}}
\]
\end{enumerate}
\end{proof}
\begin{lemma}\label{lemma:completeness_good_partition}
Let $\phi\in (0,1)$ and $q\geq 2$.
Let $\varepsilon\ll \left(\frac{\phi^2}{q^{100}}\right)^{4^{q-1}}$.
Let $\mathcal{I}=(G,q,\pi)$ be an instance of \textsc{Unique Label Cover} with $\mathsf{OPT}(\mathcal{I}) \geq 1-\varepsilon$ and $\phi_G \geq \phi$. Then there exists an $(r,\alpha,\beta,q+1)$-partitioning $C_1,\dots,C_r$ of $G'$ for $2\leq r\leq q$, where $\alpha:= \frac{f(r+1)}{30r}, \beta:=r\cdot f(r)$, and $f$ is as defined in (\ref{def:functionf}).
\end{lemma}
\begin{proof}
Let $S$ be the subset from Lemma~\ref{lem:ugcompleteness}. Note that by Lemma~\ref{lem:ugcompleteness}, it holds that $\rho_{G'}(2)\leq \max\{\phi_{G'}(S),\phi_{G'}(V\setminus S)\}=\phi_{G'}(S)\leq\frac{\varepsilon}{2}$.

On the other hand, by Lemma~\ref{lemma:ugphiq1}, we know that $\rho_{G'}(q+1)\geq \phi_{G'}(q+1)\geq \frac{\phi}{q(q+1)}$.

Now we claim that there must exist $2\leq r\leq q$ such that
\begin{eqnarray}
\rho_{G'}(r)\leq f(r) \text{ and } \rho_{G'}(r+1)> f(r+1). \label{eqn:goodpartitioning}
\end{eqnarray}
Suppose that the above is not true, i.e., for any $r$ with $2\leq r\leq q$, either $\rho_{G'}(r)> f(r)$, or $\rho_{G'}(r+1)\leq f(r+1)$. Then  since $\rho_{G'}(2)\leq \frac{\varepsilon}{2}= f(2)$ (by Fact~\ref{fact:oneaboutf}), we have that $\rho_{G'}(3)\leq f(3)$, which further implies that $\rho_{G'}(4)\leq f(4)$. Similarly, we obtain that $\rho_{G'}(5)\leq f(5)$, \ldots, $\rho_{G'}(q+1)\leq f(q+1)$. On the other hand, we note that $\rho_{G'}(q+1)\geq \frac{\phi}{q(q+1)}>f(q+1)$ by Fact~\ref{fact:oneaboutf}. This is a contradiction. Thus, there exists $r$ with $2\leq r\leq q$ such that Inequalities (\ref{eqn:goodpartitioning}) hold.

Now consider the index $r$ with the above property. Note that $\rho_{G'}(r+1)\geq f(r+1)\geq 2f(r)> 1.5\rho_{G'}(r)$.
Then by Theorem~\ref{thm:ogtrevisan}, there exists an $r$-partitioning $C_1, \dots,C_r$ of $V(G')$ that is an $(r, \alpha, \beta)$-clustering.

Finally, we note that each $C_i$ has a large volume: if one $C_i$ has volume at most $\frac{\mu_{G'}}{q+1}$, then by Lemma~\ref{lemma:ugphiq1},
\[
\phi_{G'}(C_i) > \frac{\phi}{q(q+1)},
\]
which contradicts to the fact that $\phi_{G'}(C_i) \leq \beta <\frac{\phi}{q^3}$ by Fact~\ref{fact:twoaboutf}.
\end{proof}

Now we show that one of the parts $C_1,\dots,C_r$ is close to the set $S$ from \cref{lem:ugcompleteness}.
\begin{lemma}\label{lemma:comp_onesetclosetoS}
Let $\mathcal{I}=(G,q,\pi)$ be an instance of \textsc{Unique Label Cover} satisfying the preconditions of Lemma~\ref{lemma:completeness_good_partition}. Assume further that $G$ is $d$-bounded and the minimum degree of $G$ is at least $1$. Let $C_1,\dots,C_r$ be the $(r,\alpha,\beta,q+1)$-clustering from Lemma~\ref{lemma:completeness_good_partition}. Let $S$ be the set from \cref{lem:ugcompleteness}. Then there exists $i \in \{1,\ldots,r\}$ such that
\[
\mu_{G'}(C_i\triangle S) \leq \left(\frac{4\varepsilon}{\alpha}+q^2\beta\right) d\cdot \mu_{G'}(S).
\]
\end{lemma}
\begin{proof}
For an index $i\in \{1,\dots,r\}$, we call $i$ \emph{good} if $\mu_{G'[C_i]}(C_i\cap S) \leq 7\mu_{G'[C_i]}/8$, and call it \emph{bad} otherwise.
 Now consider a good index $i$. Note that if $0<\mu_{G'[C_i]}(C_i\cap S)\leq \mu_{G'[C_i]}/2$, then $\min\{\mu_{G'[C_i]}(C_i\cap S), \mu_{G'[C_i]}(C_i\setminus S)\}=\mu_{G'[C_i]}(C_i\cap S)$; if $\mu_{G'[C_i]}/2<\mu_{G'[C_i]}(C_i\cap S)\leq 7\mu_{G'[C_i]}/8$, then $\min\{\mu_{G'[C_i]}(C_i\cap S), \mu_{G'[C_i]}(C_i\setminus S)\}\geq \frac{\alpha}{7} \mu_{G'[C_i]}(C_i\cap S)$.

Note that since $\phi_{G'[C_i]}\geq \alpha$, it holds that
\[
e_{G'}(C_i\cap S, C_i\setminus S)\geq \alpha \min\{\mu_{G'[C_i]}(C_i\cap S), \mu_{G'[C_i]}(C_i\setminus S)\}\geq  \frac{\alpha}{7}\cdot \mu_{G'[C_i]}(C_i\cap S).
\] 
Since 
\[
\sum_{i: \textrm{ good} } e_{G'}(C_i\cap S, C_i\setminus S)\leq \sum_{i=1}^r e_{G'}(C_i\cap S, C_i\setminus S)\leq e_{G'}(S,V'\setminus S)\leq \frac{\varepsilon}{2}\mu_{G'}(S), 
\] 
we have that 
\begin{eqnarray}
\frac{\alpha}{7}\cdot \sum_{i: \textrm{ good} } \mu_{G'[C_i]}(C_i\cap S) \leq \frac{\varepsilon}{2}\mu_{G'}(S).\label{eqn:goodindices}
\end{eqnarray}
Now suppose that all indices $i$ are good, then we have 
\begin{align*}
& \sum_{i: \textrm{ good} } \mu_{G'[C_i]}(C_i\cap S)
=\sum_{i=1}^r \mu_{G'[C_i]}(C_i\cap S)
\geq \sum_{i=1}^r \mu_{G'}(C_i\cap S)-\sum_{i=1}^r e_{G'}(C_i,V\setminus C_i)\\
&=\mu_{G'}(S)-r\cdot \beta\cdot \mu_{G'}
\geq \left(\frac{1}{q} - r\cdot \beta\right)\mu_{G'}
\geq \frac{1}{2q}\mu_{G'},
\end{align*}
where the last inequality follows from the fact that $\beta<\frac{1}{q^{10}}$ by Fact~\ref{fact:twoaboutf}.
Thus, 
\[\frac{\alpha}{7} \cdot \frac{\mu_{G'}}{2q} \leq \frac{\alpha}{7}  \sum_{i: \textrm{ good} } \mu_{G'[C_i]}(C_i\cap S) \leq  \frac{\varepsilon}{2}\mu_{G'}(S)\leq \frac{\varepsilon}{2}\frac{\mu_{G'}}{q},
\]
which contradicts to the fact that $\alpha>7\varepsilon$ by Fact~\ref{fact:twoaboutf}.
Therefore, there exists at least one index $i$ that is bad, i.e., $\mu_{G'[C_i]}(C_i\cap S) > 7\mu_{G'[C_i]}/8$.

Now suppose that there are at least two bad indices. First note that since $\phi_{G'}(C_i)\leq \beta$,
we have $e_{G'}(C_i,V\setminus C_i)\leq \beta \cdot \mu_{G'}(C_i)$, and thus 
\[\mu_{G'[C_i]}=\mu_{G'}(C_i)-e_{G'}(C_i,V\setminus C_i)\geq (1-\beta)\mu_{G'}(C_i).\]
Now let $i_1,i_2 \in \{1,\ldots,r\}$ be two bad indices. Note that since $C_{i_1}$ and $C_{i_2}$ are disjoint and $\mu_{G'}(C_{i_1}),\mu_{G'}(C_{i_2})\geq \frac{\mu_{G'}}{q+1}$,
\begin{align*}
& \mu_{G'}(S)
\geq \mu_{G'[C_{i_1}]}(C_{i_1}\cap S)+\mu_{G'[C_{i_2}]}(C_{i_2}\cap S)> \frac{7}{8}\mu_{G'[C_{i_1}]}+\frac{7}{8}\mu_{G'[C_{i_2}]}\\
&\geq \frac78 (1-\beta)(\mu_{G'}(C_{i_1})+\mu_{G'}(C_{i_1}))
\geq \frac{7(1-\beta)}{4}\cdot\frac{\mu_{G'}}{q+1}
>\frac{\mu_{G'}}{q}=\mu_{G'}(S),
\end{align*}
where the last inequality follows from the fact that $q\geq 2$ and $\beta\ll 1/10$. This is a contradiction. 

Therefore, there exists exactly one bad index, say $i$.
We know that
\[
\mu_{G'[C_i]}(C_i\setminus S) \leq \frac{\mu_{G'[C_i]}(C_i)}{8},
\]
and thus 
\[
\alpha \cdot \mu_{G'[C_i]}(C_i\setminus S) \leq e_{G'}(C_i\setminus S, S)\leq e_{G'}(V'\setminus S, S)\leq \frac{\varepsilon}{2} \mu_{G'}(S).
\]
By the fact that $\phi_{G'[C_i]}\geq \alpha>0$, we know that $\deg_{G'[C_i]}(v)\geq 1$ for any $v\in C_i$. Thus,

\[
|C_i\setminus S|\leq \mu_{G'[C_i]}(C_i\setminus S) \leq \frac{ \varepsilon}{2\alpha} \mu_{G'}(S)
\]
On the other hand, by Ineqaulity~\eqref{eqn:goodindices}, we have
\[
\sum_{j: \textrm{ good}} \mu_{G'[C_j]}(S\cap C_j) \leq \frac{7\varepsilon}{2\alpha} \mu_{G'}(S).
\]
Since $i$ is the only bad index, we have
\begin{align*}
	& |S\setminus C_i|\leq \mu_{G'}(S\setminus C_i)=\sum_{j: \textrm{ good}} \mu_{G'}(S\cap C_j)
	\leq \sum_{j: \textrm{ good}} \mu_{G'[C_j]}(S\cap C_j) + \sum_{j: \textrm{ good}} e_{G'}(C_j,V'\setminus C_j)\\
	&\leq \frac{7\varepsilon}{2\alpha} \mu_{G'}(S) + r\cdot \beta \cdot \mu_{G'}
	\leq \left(\frac{7\varepsilon}{2\alpha}+r\beta\cdot q\right)\mu_{G'}(S).
\end{align*}
Therefore, there exists a unique $C_i$ with 
\[
\mu_{G'}(C_i\triangle S)\leq d\cdot|C_i\triangle S|=d\cdot |C_i\setminus S|+d\cdot |S\setminus C_i|\leq \left(\frac{4\varepsilon}{\alpha}+q^2\beta\right) d\cdot \mu_{G'}(S). \qedhere
\]
\end{proof}

\subsubsection{Soundness}
In this section, we show that in the soundness, i.e., the underlying graph $G$ is an expander and $\OPT(\mathcal{I})\leq 1-\rho$, the label-extended graph $G'$ is a small set expander. More precisely, for any set of volume at most $\frac{\mu_{G'}}{q}$, its expansion is large. 

We can reuse the proof of Lemma~\ref{lemma:linsoundness_combinatorial} and obtain the following:
\begin{lemma}\label{lemma:ugsoundness_combinatorial}
	Let $\mathcal{I}=(G=(V,E),q,\pi)$ be an instance of \textsc{Unique Label Cover} with $\phi_G \geq \phi$ and $\mathsf{OPT}(\mathcal{I}) \leq 1-\rho$.
	Let $G'=(V' = V \times [q],E')$ be the label-extended graph of $\mathcal{I}$.
	Then, $\phi_{G'}(q) \geq \rho \phi/6q$.
\end{lemma}

\subsection{The Sublinear Clustering Oracle: Proof Sketch of Theorem~\ref{thm:clusteringoracle}}\label{sec:proofclusteringoracle}
The proof follows from the proof of Theorem 3 in~\cite{gluch2021spectral}. The spectral clustering oracle in~\cite{gluch2021spectral} is stated for $(k,\alpha,\beta,\Omega(1))$-clusterable graphs with $\frac{\beta\log k}{\alpha^3}\ll 1$, under a slightly different definition of being clusterable. In order to adapt their algorithm and analysis to our setting, i.e., a clustering with minimum cluster volume at least $\frac{\mu_G}{q}$, $d$-bounded graphs, and the assumption that $\frac{\beta d\cdot k^{10}}{\alpha^3}\ll 1$, we point out the differences below.

We first observe that a $(k,\alpha,\beta,q)$-clustering under our Definition~\ref{def:clusterablegraphs} is a $(k,\frac{\alpha}{d},\beta)$-clustering under the definition\footnote{In~\cite{gluch2021spectral}, given a $d$-bounded graph $G$, two subsets $S\subseteq C$, the \emph{conductance of a set $S$ within $C$} is defined to be $\phi_C(S)=\frac{e_G(S,C\setminus S)}{d|S|}$;  and the \emph{external conductance} of a set $C$ is defined to be $\phi_V(C)=\frac{e_G(C,V\setminus C)}{d|C|}$; the \emph{internal conductance} of a set $C\subseteq V$, denoted by $\phi^G(C)$, is $\min_{\emptyset \subsetneq S: \mu_G(S)\leq \frac{|C|}{2}}\phi_C(S)$ if $|C|>1$ and $0$ otherwise. A $(k,\alpha,\beta)$-clustering of $G$ is a $k$-partition $P_1,\dots,P_k$ of $V(G)$ such that for each $i \in \{1,\ldots,k\}$, $\phi^G(P_i)\geq \alpha,\phi_V(P_i)\leq \beta$. } in~\cite{gluch2021spectral}. Furthermore, since each cluster has volume at least $\frac{\mu_G}{q}$, we know that each cluster has size at least $\frac{n}{q}$.
Since $\frac{\beta \cdot d\cdot q^{10}}{\alpha^3}\ll 1$, in the following, we can assume that input graph satisfies almost all the preconditions of the main theorem (Theorem 3) in~\cite{gluch2021spectral}, with the only exception of the precondition ``for all $i,j\in[k]$, one has $\frac{|C_i|}{|C_j|}\in O(1)$'' being replaced by ``$\min_i |C_i|\geq \frac{n}{q}$''.

The clustering oracle in~\cite{gluch2021spectral} is based on a sublinear-time algorithm for approximating the dot product of the spectral embedding and a hyperplane partitioning scheme.

\paragraph{Approximating the dot product of spectral embedding.} For this subroutine, the main idea is to sample a set $S$ of vertices, and from each sampled vertex $v$, perform a number of random walks and then use the statistics of the endpoints of these walks to approximate $\langle f_u, f_v\rangle$ queries, where $f_u$ is the spectral embedding of $u$ that is defined by the bottom $k$ eigenvectors of the normalized Laplacian of $G$.
For this part to work in our setting (in which the minimum cluster size is $\Omega(\frac{n}{q})$, instead of $\Omega(\frac{n}{k})$, for some $q\geq k$), we note that whenever one applies Lemmas 4 and 5 of~\cite{gluch2021spectral}, one needs to use $\min_{i \in \{1,\ldots,k\}}|C_i|=\Omega(\frac{n}{q})$, instead of ``$\min_{i \in \{1,\ldots,k\}}|C_i|=\Omega(\frac{n}{k})$'' as before. This further leads to a change the factor $k$ in the expression of the $\ell_2$-norm upper bound in Lemma 22 in~\cite{gluch2021spectral} to $q$. This implies that one needs to replace the factor $k^8$ in the expression of the sample size $s$ in Algorithm 4 in~\cite{gluch2021spectral} by $q^8$.
The analysis of the algorithm follows by noting similar changes of dependency on $k$ to the corresponding dependency on $q$. Thus, after these changes, the running times of both preprocessing and query answering become $\widetilde{O}((\frac{q}{\alpha\beta})^{O(1)}\cdot n^{\frac12+O(\beta/\alpha^2)})$.

\paragraph{A hyperplane partitioning scheme.} Given query access to (the approximate of) the dot product $\langle f_u, f_v\rangle$, the authors of~\cite{gluch2021spectral} use an iterative approach to find a $k$-clustering. Roughly speaking, the algorithm does the following:  first sample a small number of vertices (of size $\Theta(\frac{\alpha^2}{\beta}k^4\log k)$), denoted by $S$, and then consider all possible $k$-partitioning of $S$ and find the ``right'' $k$-clustering of $S$ (i.e., the centers defined by all parts of the clustering induce a partitioning of $V(G)$ each part of which has small outer conductance). To test if a $k$-partitioning of $S$ is ``right'', the algorithm uses the dot product of spectral embedding oracle to
test if a vertex belongs to some cluster $C$ induced by a specific center and test if such a cluster $C$ has small outer conductance, which in turn is sketched in the paragraph before Lemma~\ref{lem:testouterconductance}.
Then one can use the ``right'' partitioning of $S$ to answer which cluster a queried vertex $x$ belongs to. The algorithm in~\cite{gluch2021spectral} uses a more involved iterative approach that finds the good clustering in stages to deal with some technical challenges, which we omit the details. To adjust their algorithm and analysis to our setting, we note that one only needs to increase the sample size (i.e., $|S|$) to be $\Theta(\frac{\alpha^2}{\beta}q^4\log q)$ so that sufficiently many vertices in the smallest cluster are sampled. In the analysis, we need to replace some ``$\frac{\beta}{\alpha^2}$'' by ``$\frac{\beta q}{\alpha^2}$'', which occurs due to the fact that in our setting, $\frac{\max_{i \in \{1,\ldots,k\}}|C_i|}{\min_{i \in \{1,\ldots,k\}}|C_i|}=O(q)$ rather than $O(1)$ as in~\cite{gluch2021spectral}. Finally, we note that the analysis still holds for the modified algorithm, as we have assumed $\frac{\beta q^{10}}{\alpha^3}\ll 1$.

\subsection{Proof of Theorem~\ref{thm:unique-game}}

\paragraph{Completeness} In this case, $\OPT(\II)\geq 1-\varepsilon$. By Lemma~\ref{lem:ugcompleteness}, there exists a subset $S$ in the label-extended graph $G':=G_\II$ of volume $\mu_{G'}/q$ and $\phi_{G'}(S)\leq \varepsilon/2$. By Lemma~\ref{lemma:completeness_good_partition}, there exists an $(r,\alpha,\beta,q+1)$-clustering $C_1,\ldots,C_r$ of the label-extended graph $G':=G_\II$, for some $2\leq r\leq q$, where 
where $\alpha= \frac{f(r+1)}{30r}, \beta=r\cdot f(r)$, and $f$ is as defined in (\ref{def:functionf}).
By Lemma~\ref{lemma:comp_onesetclosetoS}, one of the sets $C_1,\ldots,C_r$, say $C_1$, satisfies that
\begin{align}
\mu_{G'}(C_1\triangle S) \leq \left(\frac{4\varepsilon}{\alpha}+q^2\beta\right) \mu_{G'}(S).\label{eqn:symmetric1}
\end{align} 
Note that $C_1,\ldots,C_r$ is a $(r,\alpha, \beta,q+1)$-clustering of $G'$.

Note that by Theorem~\ref{thm:clusteringoracle}, the spectral clustering oracle $\mathcal{O}$ with input $G'$ and parameters $\alpha,\beta,r,q+1$ provides consistent query access to a partition $(\widehat{C}_1,\ldots,\widehat{C}_r)$ such that with probability at least $0.9$, it holds that for some permutation $\tau:[r]\to [r]$ and for any $i \in \{1,\ldots,r\}$, we have
\begin{align}
\mu_{G'}(C_i\triangle \widehat{C}_{\tau(i)})\leq O_d\left(\frac{\beta\cdot q^{10}}{\alpha^3}\right)\mu_{G'}(C_i)\label{eqn:symmetric2}
\end{align} Note that it implies that for any $i \in \{1,\ldots,r\}$, we have
\[\mu_{G'}(\widehat{C}_i)\geq \left(1-O_d\left(\frac{\beta\cdot q^{10}}{\alpha^3}\right)\right)\mu_{G'}(C_i)\geq \frac{0.9\mu_{G'}}{q+1}.\]
Thus, $|\widehat{C_i}|\geq \frac{0.9\mu_{G'}}{(q+1)d}\geq \frac{|V(G')|}{10dq}$. Furthermore, with high probability, the estimator $s_i$ of the volume $\mu_{G'}(\widehat{C}_i)$ satisfies that $s_i\geq \frac{\mu_{G'}}{2(q+1)}\geq \frac{nq}{2(q+1)}$, for each $i \in \{1,\ldots,r\}$.

Furthermore, by Inequalities~\eqref{eqn:symmetric1} and~\eqref{eqn:symmetric2} and the fact that $\mu_{G'}(S)=\frac{\mu_{G'}}{q}$, we have
\[
\mu_{G'}(C_1)\in \left[\left(1-\frac{4\varepsilon}{\alpha}-q^2\beta\right)\frac{\mu_{G'}}{q},\left(1+\frac{4\varepsilon}{\alpha}+q^2\beta\right)\frac{\mu_{G'}}{q}\right]
\]  and
\[
\mu_{G'}(\widehat{C}_{\tau(1)})\in \left[\left(1-O_d\left(\frac{\beta\cdot q^{10}}{\alpha^3}\right)\right){\mu_{G'}(C_1)},\left(1+O_d\left(\frac{\beta\cdot q^{10}}{\alpha^3}\right)\right){\mu_{G'}(C_1)}\right].
\]

If we set $\xi:=O_d(\frac{\varepsilon}{\alpha}+\frac{\beta\cdot q^{10}}{\alpha^3})=O_d(\frac{\beta q^{10}}{\alpha^3})=O_d(\frac{q^{85q}\cdot \varepsilon^{4^{1-r}}}{\phi^{\frac{2r-1}{q-1}}})$, where the last equation follows from Fact~\ref{fact:twoaboutf}, then \[\mu_{G'}(\widehat{C}_{\tau(1)})\in \left[\left(1-\frac{\xi}{4}\right)\frac{{\mu_{G'}}}{q}, \left(1+\frac{\xi}{4}\right)\frac{{\mu_{G'}}}{q}\right]=\left[\left(1-\frac{\xi}{4}\right)\mu_{G}, \left(1+\frac{\xi}{4}\right)\mu_{G}\right],\]
where the last equation holds as $\mu_{G'}=\mu_G\cdot q$.

In Line~\ref{alg:estimateofvolumeGprime} of Algorithm~\ref{alg:uniquegames}, since we sampled $O(\frac{dq\log n}{\xi_0^2})$ vertices, the degree of each vertex in $G$ is in $\{1,\ldots,d\}$ and $\xi_0=O(\frac{ q^{50}\cdot \varepsilon^{4^{1-r}}}{\phi^{\frac{2r-1}{q-1}}})\leq \xi=O_d(\frac{\beta q^{10}}{\alpha^3})$ which in turn follows from Fact~\ref{fact:twoaboutf}.
By the Chernoff bound, with high probability, the estimate $x$ satisfies that
\[
x\in \left[\left(1-\frac{\xi}{4}\right)\mu_{G},\left(1+\frac{\xi}{4}\right)\mu_{G} \right].
\] 

Now we consider the for-loop Algorithm in~\ref{alg:uniquegames} when $r$ is equal to the number of partitions guaranteed by Lemma~\ref{lemma:completeness_good_partition}.
In Line~\ref{alg:definefisi} of Algorithm~\ref{alg:uniquegames}, since we sampled $s=O(qd\log n)$ vertices, the degree of each vertex is in $\{1,\ldots,d\}$ and $\mu_{G'}(\widehat{C}_i)\geq \frac{0.9\mu_{G'}}{q+1}=\frac{0.9q\cdot\mu_{G}}{q+1}$, with high probability, the estimates $s_i$ satisfy that
\[s_i\geq \frac{q\cdot \mu_{G}}{2(q+1)}, \textrm{ for any $i \in \{1,\ldots,r\}$, and } s_{\tau(1)}\in \left[\left(1-\frac{\xi}{2}\right)  \mu_{G},\left(1+\frac{\xi}{4}\right) \mu_{G}\right]
\]
The above implies that 
\[
s_i\geq \frac{x\cdot q}{4(q+1)} \textrm{ for any $i \in \{1,\ldots,r\}$, and } s_{\tau(1)}\in \left[\left(1-\xi\right)  x,\left(1+\xi\right) x\right]
\]

Now note that 
\[\mu_{G'}(\widehat{C}_{\tau(1)}\triangle S)\leq \mu_{G'}(\widehat{C}_{\tau(1)}\triangle C_1) + \mu_{G'}(C_1\triangle S)\leq \frac{\xi}{4}\mu_{G'}(S).\]

Since $\phi_{G'}(S)\leq \frac{\varepsilon}{2}$, we have that 
\[
\phi_{G'}(\widehat{C}_{\tau(1)})\leq \frac{e_{G'}(S,V'\setminus S)+e_{G'}(C_1\setminus S,V'\setminus C_1)}{(1-\frac{\xi}{4})\mu_{G'}(S)}\leq\frac{\frac{\varepsilon}{2}\mu_{G'}(S)+\frac{\xi}{4}\mu_{G'}(S)}{(1-\frac{\xi}{4})\mu_{G'}(S)} \leq \varepsilon+\xi.
\]

By Lemma~\ref{lem:testouterconductance}, \textsc{TestOuterConductance} will output an estimate $\eta$ such that $\eta=O(\varepsilon+\xi+ \frac{\beta}{\alpha^2})=O(\xi)=O_d(\frac{q^{85q}\cdot \varepsilon^{4^{1-r}}}{\phi^{\frac{2r-1}{q-1}}})$. 
Thus, the instance will be accepted.

\paragraph{Soundness} In this case, $\OPT(\II)\leq 1-\rho$. By Lemma~\ref{lemma:ugsoundness_combinatorial},  any set of volume at most $\mu_{G'}/q=\mu_G$ has conductance at least $\frac{\rho\phi}{6q}$. This further implies that any set of volume at most $(1+2\xi)\mu_G$
has conductance at least
\[
\frac{\frac{\rho \phi}{6q} \mu_G - 2\xi  \cdot\mu_G}{(1+2\xi)\mu_G} \geq \Omega\left(\frac{\rho \phi}{q}\right).
\]

As with the completeness case, with high probability, the estimate $x$ (defined in Line~\ref{alg:estimateofvolumeGprime} of Algorithm~\ref{alg:uniquegames}) satisfies that $x\in[(1-\frac{\xi}{4})\mu_{G},(1+\frac{\xi}{4})\mu_G]$. Assume that
\begin{eqnarray}
	(*) \textrm{ for each $j\leq r$, $s_j\geq \frac{x\cdot q}{4(q+1)}$, and there exist $i$ such that $s_i\in [(1-\xi) x, (1+\xi) x]$},
\end{eqnarray}
as otherwise, the instance will be rejected. Note that we can further assume that the corresponding sets $\widehat{C}_j$ has size at least $\frac{\mu_{G'}}{8q}=\frac{\mu_{G}}{8}>\frac{|V(G')|}{10dq}$ for each $j\leq r$ and the set $\widehat{C}_i$ has volume at most $(1+2\xi)x$, as otherwise, with high probability, the assumption (*) does not hold and the instance will be rejected.

Since $\widehat{C}_i$ has volume at most $(1+2\xi)n$, by the above argument, we know that the outer conductance of $\widehat{C}_i$ is at least $\Omega(\frac{\rho\phi}{q})$.
By Lemma~\ref{lem:testouterconductance},  \textsc{TestOuterConductance} will output an estimation at least
$\Omega(\frac{\rho\phi}{dq}- \frac{\beta}{\alpha^2})\gg O_d(\frac{q^{85q}\cdot \varepsilon^{4^{1-r}}}{\phi^{\frac{2r-1}{q-1}}})$, as $\rho=\Omega_d(q^{86q}\cdot\varepsilon^{4^{1-q}}/\phi^4)$.
 Thus the instance will be rejected. 

\paragraph{Running time} Note that the running time (and the query complexity) of the algorithm are dominated by the time (and the number of queries) of invoking the spectral clustering oracle $\mathcal{O}$ and the subroutine \textsc{TestOuterConductance}. For any $r\leq q$ and the corresponding $\alpha,\beta$, both times are $ \widetilde{O}(\poly(dq/\alpha\beta)^{O(1)}\cdot 2^{O((\alpha^2/\beta) \cdot q^{100})}\cdot n^{\frac12+O(\beta/\alpha^2)})$.
By Fact~\ref{fact:twoaboutf}, it holds that $\frac{\alpha^2}{\beta}=O(\frac{\phi^{2/(q-1)}\varepsilon^{-1/2}}{q^{40}})$, and $\frac{\beta}{\alpha^2}=O(\frac{q^{40q}\cdot \varepsilon^{4^{1.5-q}}}{\phi^{\frac{q}{q-1}}})$, and the total query complexity and running time are thus $\widetilde{O}_d(2^{O(q^{60}\cdot\phi^{2/(q-1)}\cdot \varepsilon^{-1/2})}\cdot n^{\frac12+O(q^{40q}\cdot \varepsilon^{4^{1.5-q}}\cdot \phi^{-\frac{q}{q-1}})})=\widetilde{O}_d(2^{q^{O(1)}\cdot\phi^{1/q}\cdot \varepsilon^{-1/2}}\cdot n^{\frac12+q^{O(q)}\cdot \varepsilon^{4^{1.5-q}}\cdot \phi^{-2}})$. This proves the query complexity and running time of the algorithm.


\section{Testing $3$-Colorability is Hard on Expander Graphs}\label{sec:3color}
In this section, we prove Theorem~\ref{thm:3-colorability}.

Our lower bound uses the same construction as the one given by Bogdanov et al.~\cite{BOT02:test} for showing that testing $3$-colorability requires $\Omega(n)$ queries in the bounded-degree graph model. Their lower bound was obtained by first giving a lower bound of $\Omega(n)$ queries for testing the satisfiability of \textsc{E3SAT}, and then showing a reduction from \textsc{E3SAT} to $3$-colorability. Here, we show that the graph family obtained from their reduction is a family of expander graphs.

Let us first recall the reduction given in~\cite{BOT02:test}. We will make use of the following notion of expander graphs. 

\begin{definition}
Let $d\geq 8$ be some constant. A graph $G=(V,E)$ is an $(n,d)$-expander if $|V|=n$, it is $d$-regular and if for every subset $S\subseteq V$ with $|S|\leq \frac{|V|}{2}$, $|\Gamma_G(S)|\geq |S|$, where $\Gamma_G(S)$ is the set of neighbors of $S$ in $G$.
\end{definition}
For some fixed constant $d\geq 8$, it is known that a family of infinite number of $(n,d)$-expanders $G_n$ can be explicitly constructed~\cite{margulis1973explicit,gabber1981explicit}.

\paragraph{Construction} Given a $3$-CNF $f$ such that each literal appears in at most $k$ clauses, we construct a graph $\psi(f)$. Keep it in mind that we would like to color the vertices in the graph using three colors. We define the gadgets, vertex set, and edge set as follows.
\begin{itemize}
\item Vertex set:
\begin{itemize}
	\item color class vertices: we introduce three classes of color vertices: $D_i$, $T_i$, $F_i$, where $1\leq i\leq 2kn$. The colors of vertices $D_i$ will all correspond to ``dummy'' color, $T_i$ to ``true'' color, and $F_i$ to ``false'' color.
	\item literal vertices: for each variable $x_i$ in $f$, we introduce $2k$ literal vertices $x_i^1,\dots,x_i^k,\overline{x_i^1}, \dots, \overline{x_i^k}$
	\item additional vertices (called $A$-vertices): those are vertices that belong to some gadget, which in turn is defined between color class vertices or literal vertices. 
\end{itemize}

\item The gadgets: 
\begin{itemize}
\item equality gadget: for any two vertices (either color class or literal vertices) that are supposed to have the same color, an equality gadget is introduced. See Figure~\ref{fig:H}(a).
\item clause gadget: for each clause, we introduce a gadget on the literals appearing in the clause, that allows any coloring of the literal vertices with ``true'' or ``false'' other than the coloring which corresponds to an assignment where all literals are false (and the clause goes unsatisfied). See Figure~\ref{fig:H}(b).
\end{itemize}

\begin{figure*}[htb]
	\centering
	\begin{subfigure}{0.4\linewidth}
		\centering
		\begin{tikzpicture}
			\tikzstyle{edge}=[line width=1.3]
			\tikzstyle{vertex}=[draw,circle,fill=black,inner sep=0pt, minimum width=9pt]
			\tikzstyle{node}=[draw,circle,fill=white,inner sep=0pt, minimum width=9pt]
			\def \dist {1.5}
			\node[node] (8) at (0,-2*\dist) {};	
			\node[node] (2) at (0,0) {};
			\node[vertex] (3) at (-1*\dist,-1*\dist) {};
			\node[vertex] (7) at (1*\dist,-1*\dist) {};
			\draw[edge] (2)--(8);
			\draw[edge] (3)--(2);
			\draw[edge] (2)--(7);
			\draw[edge] (7)--(2);
			\draw[edge] (3)--(8);
			\draw[edge] (7)--(8);
			
			\node at (3)[above=4pt]{$y_{1}$};
			\node at (7)[above=4pt]{$y_{2}$};
		\end{tikzpicture}
		\caption{$y_1=y_2$}
\end{subfigure}
\begin{subfigure}{0.4\linewidth}
	\centering
	\begin{tikzpicture}
		\tikzstyle{edge}=[line width=1.3]
		\tikzstyle{vertex}=[draw,circle,fill=black,inner sep=0pt, minimum width=9pt]
		\tikzstyle{node}=[draw,circle,fill=white,inner sep=0pt, minimum width=9pt]
		\def \dist {1.5}
		\node[vertex] (1) at (0,1*\dist) {};	
		
		\node[node] (2) at (0,0*\dist) {};	
		
		\node[vertex] (3) at (-1,0*\dist) {};
		
		\node[node] (4) at (0,-1*\dist) {};
		
		\node[node] (5) at (-1.5,-2*\dist) {};	
		
		\node[node] (6) at (1.5,-2*\dist) {};	
		
		\node[vertex] (11) at (-2,1*\dist) {};
		
		\node[node] (22) at (-2,0*\dist) {};	
		
		\node[vertex] (33) at (-3,0*\dist) {};	
		
		\node[vertex] (111) at (2,1*\dist) {};
		
		\node[node] (222) at (2,0*\dist) {};	
		
		\node[vertex] (333) at (1,0*\dist) {};

		\draw[edge] (1)--(2);
		\draw[edge] (3)--(2);
		\draw[edge] (2)--(4);
		
		\draw[edge] (4)--(5);
		\draw[edge] (4)--(6);
		\draw[edge] (5)--(6);	
		
		\draw[edge] (11)--(22);
		\draw[edge] (33)--(22);
		
		\draw[edge] (111)--(222);
		\draw[edge] (333)--(222);
		
		\draw[edge] (5)--(22);
		\draw[edge] (222)--(6);
		
		\node at (1)[above=4pt]{$x_{2}$};
		\node at (11)[above=4pt]{$x_{1}$};
		\node at (111)[above=4pt]{$\overline{x_{3}}$};
		
		\node at (3)[above=4pt]{$T_{i_2}$};
		\node at (33)[above=4pt]{$T_{i_1}$};
		\node at (333)[above=4pt]{$T_{i_3}$};
	\end{tikzpicture}
	\caption{$x_1\vee x_2\vee\overline{x_3}$}
\end{subfigure}
\caption{A black node represents a color class vertex or a literal vertex; all the white nodes represent additional vertices ($A$-vertices) in the gadgets.
}\label{fig:H}
\end{figure*}
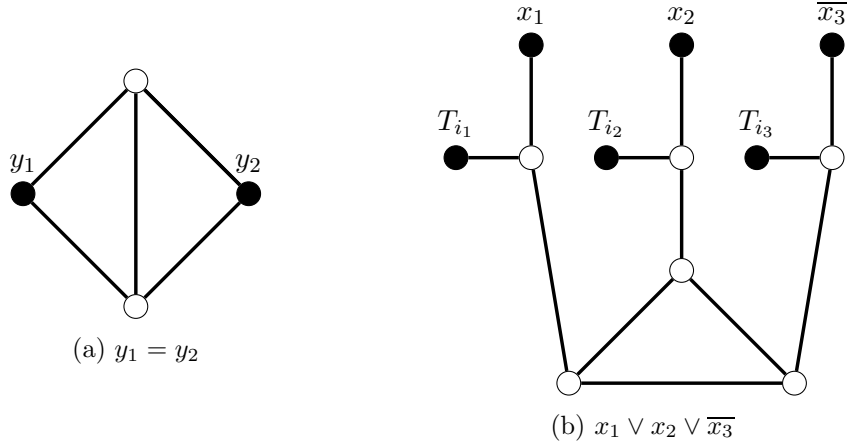

\item Edge set:
\begin{itemize}
\item We first add a $(2kn,d)$-expander graph on the set of vertices $\{D_i\}_{i=1,\dots,2kn}$, for some constant $d>0$. Similarly, we add a $(2kn,d)$-expander graphs on the set $\{T_i\}_{i=1,\dots,2kn}$, $\{F_i\}_{i=1,\dots,2kn}$ respectively.
\item add equality gadgets between literal vertices $x_i^j,x_i^{j'}$ for all $1\leq i\neq j\leq k$ (similarly for $\overline{x_i^j},\overline{x_i^{j'}}$). This is to ensure that for any variable, its literal vertices should be colored consistently. 
\item add edges $(x_i^j, \overline{x_i^j})$ for all $i,j$, as only one of $x_i,\overline{x_i}$ can be true.
\item fix some one-to-one correspondence between the literal vertices and the color class vertices. Connect each literal vertex to its corresponding vertex $D_i$, since it will be colored with only ``true'' or ``false''. 
\end{itemize}

\end{itemize}

The resulting graph is given in Figure~\ref{fig:construction}.

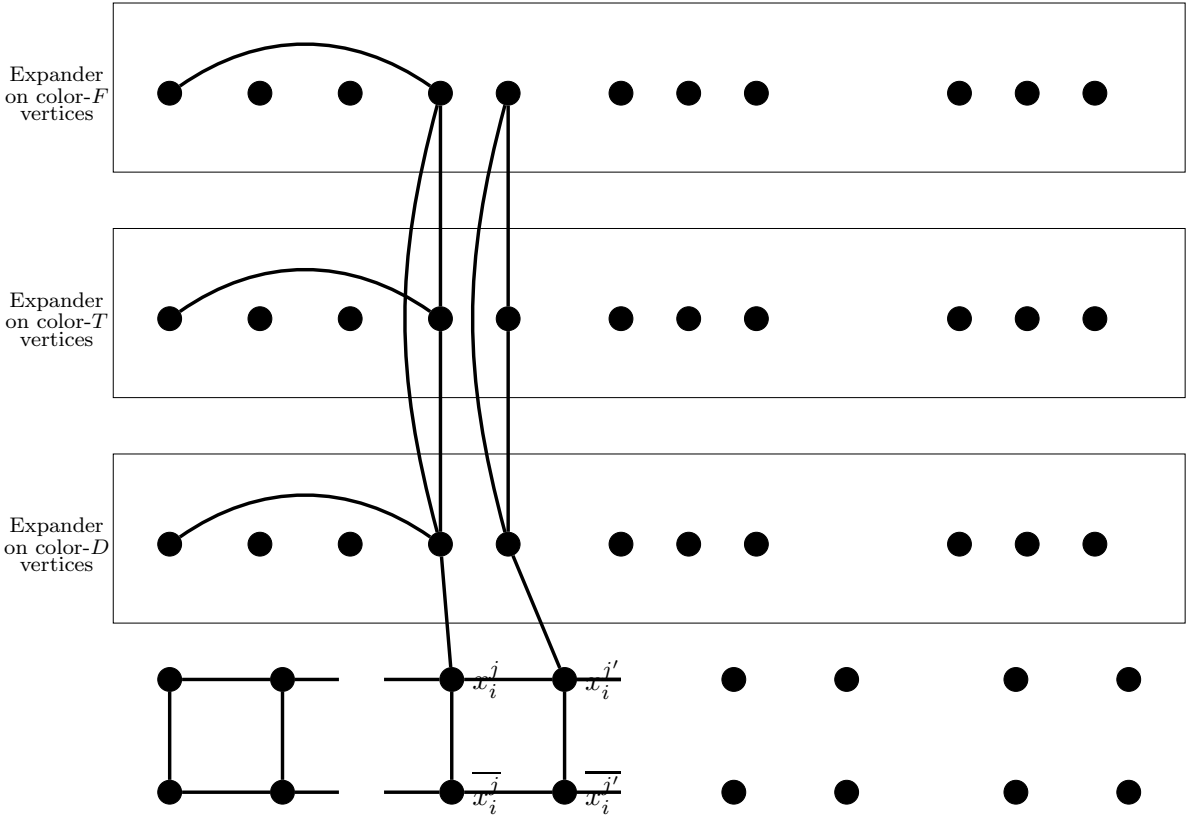
\begin{figure*}[htb]
	\centering
	\begin{tikzpicture}
		\tikzstyle{edge}=[line width=1.3]
		\tikzstyle{vertex}=[draw,circle,fill=black,inner sep=0pt, minimum width=9pt]
		\tikzstyle{node}=[draw,circle,fill=white,inner sep=0pt, minimum width=9pt]
		\def \dist {1.5}
		\node[vertex] (1) at (0,1*\dist) {};
		\node[vertex] (2) at (1*\dist,1*\dist) {};
		\node[vertex] (3) at (2.5*\dist,1*\dist) {};
		\node[vertex] (4) at (3.5*\dist,1*\dist) {};
		\node[vertex] (5) at (5*\dist,1*\dist) {};
		\node[vertex] (6) at (6*\dist,1*\dist) {};
		\node[vertex] (7) at (7.5*\dist,1*\dist) {};
		\node[vertex] (8) at (8.5*\dist,1*\dist) {};

		\node[vertex] (9) at (0,2*\dist) {};
		\node[vertex] (10) at (1*\dist,2*\dist) {};
		\node[vertex] (11) at (2.5*\dist,2*\dist) {};
		\node[vertex] (12) at (3.5*\dist,2*\dist) {};
		\node[vertex] (13) at (5*\dist,2*\dist) {};
		\node[vertex] (14) at (6*\dist,2*\dist) {};
		\node[vertex] (15) at (7.5*\dist,2*\dist) {};
		\node[vertex] (16) at (8.5*\dist,2*\dist) {};

		\node[vertex] (17) at (0,3.2*\dist) {};
		\node[vertex] (18) at (0.8*\dist,3.2*\dist) {};
		\node[vertex] (19) at (1.6*\dist,3.2*\dist) {};
		\node[vertex] (20) at (2.4*\dist,3.2*\dist) {};
		\node[vertex] (21) at (3*\dist,3.2*\dist) {};
		\node[vertex] (22) at (4*\dist,3.2*\dist) {};
		\node[vertex] (23) at (4.6*\dist,3.2*\dist) {};
		\node[vertex] (24) at (5.2*\dist,3.2*\dist) {};
		\node[vertex] (25) at (7*\dist,3.2*\dist) {};
		\node[vertex] (26) at (7.6*\dist,3.2*\dist) {};
		\node[vertex] (27) at (8.2*\dist,3.2*\dist) {};

			\node[vertex] (28) at (0,5.2*\dist) {};
		\node[vertex] (29) at (0.8*\dist,5.2*\dist) {};
		\node[vertex] (30) at (1.6*\dist,5.2*\dist) {};
		\node[vertex] (31) at (2.4*\dist,5.2*\dist) {};
		\node[vertex] (32) at (3*\dist,5.2*\dist) {};
		\node[vertex] (33) at (4*\dist,5.2*\dist) {};
		\node[vertex] (34) at (4.6*\dist,5.2*\dist) {};
		\node[vertex] (35) at (5.2*\dist,5.2*\dist) {};
		\node[vertex] (36) at (7*\dist,5.2*\dist) {};
		\node[vertex] (37) at (7.6*\dist,5.2*\dist) {};
		\node[vertex] (38) at (8.2*\dist,5.2*\dist) {};

			\node[vertex] (39) at (0,7.2*\dist) {};
		\node[vertex] (40) at (0.8*\dist,7.2*\dist) {};
		\node[vertex] (41) at (1.6*\dist,7.2*\dist) {};
		\node[vertex] (42) at (2.4*\dist,7.2*\dist) {};
		\node[vertex] (43) at (3*\dist,7.2*\dist) {};
		\node[vertex] (44) at (4*\dist,7.2*\dist) {};
		\node[vertex] (45) at (4.6*\dist,7.2*\dist) {};
		\node[vertex] (46) at (5.2*\dist,7.2*\dist) {};
		\node[vertex] (47) at (7*\dist,7.2*\dist) {};
		\node[vertex] (48) at (7.6*\dist,7.2*\dist) {};
		\node[vertex] (49) at (8.2*\dist,7.2*\dist) {};
		
		\draw[very thin] (-0.5*\dist,2.5*\dist) rectangle (9*\dist,4*\dist);
		
		\draw[very thin] (-0.5*\dist,4.5*\dist) rectangle (9*\dist,6*\dist);
		
		\draw[very thin] (-0.5*\dist,6.5*\dist) rectangle (9*\dist,8*\dist);
	
		\draw[edge] (1)--(2);
	\draw[edge] (3)--(4);
		\draw[edge] (11)--(12);
  	\draw[edge] (9)--(10);
  	
		\draw[edge] (3)--(11);
\draw[edge] (4)--(12);
		\draw[edge] (1)--(9);
\draw[edge] (2)--(10);
\draw[edge] (2)--(1.5*\dist,1*\dist);
\draw[edge] (10)--(1.5*\dist,2*\dist);
\draw[edge] (3)--(1.9*\dist,1*\dist);
\draw[edge] (11)--(1.9*\dist,2*\dist);
\draw[edge] (4)--(4*\dist,1*\dist);
\draw[edge] (12)--(4*\dist,2*\dist);

\draw[edge] (12)--(4*\dist,2*\dist);

\draw[edge] (11)--(20)--(31)--(42);
\draw[edge] (12)--(21)--(32)--(43);
\draw[edge] (20) to[out=105,in=-105] (42);
\draw[edge] (21) to[out=105,in=-105] (43);

\draw[edge] (17) to[out=35,in=145] (20);
\draw[edge] (28) to[out=35,in=145] (31);
\draw[edge] (39) to[out=35,in=145] (42);

	\node at (11)[right=4pt]{$x_{i}^j$};
\node at (12)[right=4pt]{$x_{i}^{j'}$};

	\node at (3)[right=4pt]{$\overline{x_{i}^j}$};
\node at (4)[right=4pt]{$\overline{x_{i}^{j'}}$};

\node at (-1*\dist,3.2*\dist){$\substack{\text{Expander} \\\text{on color-$D$} \\ \text{vertices}}$};

\node at (-1*\dist,5.2*\dist){$\substack{\text{Expander} \\\text{on color-$T$} \\ \text{vertices}}$};

\node at (-1*\dist,7.2*\dist){$\substack{\text{Expander} \\\text{on color-$F$} \\ \text{vertices}}$};	
	
	\end{tikzpicture}

\caption{Edges of the constructed graph $\psi(f)$. Each box corresponds to an expander that is defined on some color class vertices of the same color, i.e., either $D$, $T$ or $F$; furthermore, each edge in the expander represensts an equality gadget. Clause gadgets and equality gadgets between literal vertices are not shown.  }\label{fig:construction}
\end{figure*}

\paragraph{Properties of the construction} It was shown in~\cite{BOT02:test} that the above construction is a local reduction in the sense that
\begin{itemize}
\item if $f$ is satisfiable, then $\psi(f)$ is $3$-colorable;
\item if $f$ is $\varepsilon$-far from being satisfiable, then $\psi(f)$ is $\varepsilon$-far from being $3$-colorable;
\item the answer to a query to $\psi(f)$ can be computed by making $O(1)$ queries to $f$. 
\end{itemize}
It was further shown that testing the satisfiability of $f$ requires $\Omega(n)$ queries, which implies an $\Omega(n)$ lower bound for testing $3$-colorability. Now we prove the following theorem, which directly implies Theorem~\ref{thm:3-colorability}.
\begin{theorem}
Let $f$ be an instance of $3$-CNF\@.
Let $G=\psi(f)$ be the graph constructed as above. Then it holds that $G$ is $d'$-bounded graph with $\phi(G)\geq \phi$, for some constants $d'$ and $\phi$.
\end{theorem}
\begin{proof}
Let $V=V(G)$. Note that by construction, $n<|V|\leq cn$ for some constant $c>1$, where $n$ is the number of variables in $f$. This can be seen as follows. Let $G'$ be the subgraph induced by all color class vertices (and the relevant $A$-vertices in the corresponding equality gadgets), i.e., the subgraph corresponding to the top three layers in Figure~\ref{fig:construction}. Note that since each expander has exactly $d\cdot 2kn / 2= dkn$ edges and each edge introduces two $A$-vertices, we know that $|V(G')|=3\cdot 2kn + 2\cdot 3\cdot dkn= 6kn (d+1)$. On the other hand, note that there are $2kn$ literal vertices, and at most $2 k^2\cdot n$ equality gadgets involving these vertices. Furthermore, there are at most $2kn$ clause gadgets. Since each equality gadget introduces $2$ $A$-vertices, and each clause gadget introduces $6$ $A$-vertices, we know there are at most $2kn + 2k^2\cdot n \cdot 2+ 2kn\cdot 6\leq (6k^2 +12k)n$ vertices in the bottom layer.

Note also that $G'$ is an expander. This is true, as each of the three layers is an expander and we add a perfect matching between each pair of expanders. 

Consider an arbitrary subset $S\subseteq V$ with $|S|\leq \frac{|V|}{2}$. Let $S_1$ denote the subset of $S$ that contains all color class vertices and the $A$-vertices in the equality gadgets in $G'$. Let $S_L$ be the subset of $S$ that contains all literal vertices. Let $S_2$ be the subset of $S$ of remaining vertices, which are $A$-vertices involving literal vertices in the clause gadgets and equality gadgets. Note that $S=S_1\cup S_L\cup S_2$.

By construction, we have the following properties:
\begin{itemize}
	\item the set $S_L$ has at least $|S_L|$ neighbors that are color-$D$ vertices.
	\item the set $S_2$ has either at least $0.99|S_2|$ neighbors in $S$ that are $A$-vertices, or at least $\Omega(|S_2|)$ neighbors that are color class vertices or literal vertices.
\end{itemize}

Let $S_0$ be the largest subset of $S$ among $\{S_1,S_L,S_2\}$. Note that $|S_0|\geq \frac{|S|}{3}$.
Note that there cannot be edges between $S_1$ and $S_2$, as all vertices in $S_2$ are $A$-vertices that appear at the bottom layer which can only connect to $S_L$. Consider the set $\Gamma_G(S_0)$, i.e., the set of all neighbors of $S_0$ in $G$. Let $\gamma\in (0,1)$ be a sufficiently small constant. We consider the following two cases:

\begin{itemize}
\item  Case 1: the number of neighbors of $S_0$ outside of $S$ is at least $\gamma|S_0|$, i.e., $e_G(S_0,\overline{S})\geq \gamma|S_0|$. Then $\phi_G(S)\geq \frac{e_G(S_0,\overline{S})}{d |S|}=\Omega(1/d)$.
\item Case 2: the number of neighbors of $S_0$ inside $S$ is at least $(1-\gamma)|S_0|$. 
\begin{itemize}
\item if $S_0=S_L$, then $|S_L|\geq\frac{|S|}{3}$. Furthermore, $|S_1|\leq |S_L|\leq 2kn$. Note that $S_1$ is a subset of $G'$, which consists of at least $6kdn$ vertices. Since $G'$ is an expander and $|S_1|\leq \frac{|V(G')|}{2}$, we know that the number of neighbors of $S_1$ in $\overline{S}$ is at least $\Omega(|S_1|)$. Thus, if $|S_1|\geq \gamma |S|$, then $\phi_G(S)\geq \Omega(1/d)$. If $|S_1|<\gamma|S|$, then $S_L$ has at least $|S_L|-|S_1|$ $D$-neighbors in $\overline{S}$, which also gives that $\phi_G(S)\geq \frac{(\frac13-\gamma)|S|}{d |S|}=\Omega(1/d)$. 

\item if $S_0=S_2$, then $|S_2|\geq \frac{|S|}{3}$. Note that each $A$-vertex in $S_1$ connects to at most $1$ vertex (of color $T$) in $G'$, and to at most $1$ literal vertex in $S_L$.
Note that for any vertex $v\in S_2$, it either connects to a $T$-vertex or a literal vertex in $V_2\setminus S_2$, or it connects to another vertex $w\in S_2$ that connects to a $T$-vertex or a literal vertex in $V_2\setminus S_2$. Therefore, the number of neighbors of $S_2$ that are either literal vertices or $T$ vertices is at least $\frac{|S_2|}{2}$. If there are $\frac{|S_2|}{4}$ $T$-neighbors of $S_2$, then either at least $\frac{|S_2|}{8}$ such $T$-neighbors are outside $S$, which implies $e_G(S_2,\overline{S})=\Omega(|S_2|)$ or at least $\frac{|S_2|}{8}$ such $T$-neighbors are in $S_1$, which in turn has $\Omega(|S_1|)=\Omega(|S_2|)$ neighbors outside $S$. That is, in both sub-cases, $e_G(S,\overline{S})=\Omega(|S|)$, which implies that $\phi_G(S)=\Omega(1/d)$. 

\item if $S_0=S_1$, then $|S_1|\geq \frac{|S|}{3}$. Note that each vertex $v\in S_1$ is either connected to a literal vertex or is connected to an $A$-vertex in the bottom layer, or is connected to a vertex in $G'$.  Then similar to the analysis of the above analysis, we can bound that $e_G(S,\overline{S})=\Omega(|S|)$, which implies that $\phi_G(S)=\Omega(1/d)$.
\end{itemize}

\end{itemize}
That is, in both cases, $\phi_G(S)=\Omega(1)$. This finishes the proof of the theorem.
\end{proof}

\bibliographystyle{abbrv}
\bibliography{SublinearUniGames}

\appendix

\section{Proof of Theorem~\ref{thm:distribution}}\label{app:deferred_proofs}
The proof of Theorem~\ref{thm:distribution} is based on an easy modification of the
$\ell_2$-norm testing algorithms for two distributions given in~\cite{CDVV14:optimal}.
We give the proof here for the sake of completeness.
We first show the following lemma.
Let $\Poi(r)$ denote the Poisson distribution with parameter $r$.

\begin{lemma}\label{lemma:dist}
  Let $r > 0$, $G=(V,E)$ be a graph, and let $\p,\q$ be two distributions over $V$.
  Then, there exists an algorithm with sample and time complexities $O(r)$ such that, for any $\xi > 0$, it outputs an estimate of $\|(\p-\q)D^{-1/2}\|_2^2$ with an additive error of $\xi$ with probability at least $3/4$, provided that $r\geq C\frac{\sqrt{b}}{\xi}$,
  where $b = \max\left\{\|\p D^{-1/2}\|_2^2, \|\q D^{-1/2}\|_2^2\right\}$ and $C > 0$ is an absolute constant.
\end{lemma}

Note that Theorem~\ref{thm:distribution} directly follows from the above Lemma, by the standard trick of boosting the success probability by repetition and the fact that $\|(\p-\q)D^{-\frac12}\|_4^2\leq \|(\p-\q)D^{-\frac12}\|_2^2$.

\begin{algorithm}[t!]
  \caption{\textsc{$l_2$-DifferenceTest}}\label{alg:difference-test}
  \Input{$r > 0$, query access to a graph $G=(V,E)$, and sampling access to two distributions $\p$ and $\q$ over $V$.}
  Draw $k$ from $\Poi(r)$\;
  \If{$k > 8r$}{
    Abort.\label{line:difference-test-abort}
  }
  Draw $k$ samples from each distribution $\p$ and $\q$\;
  Let $X_v,Y_v\;(v \in V)$ denote the number of occurrences of $v$ in the samples from $\p$ and $\q$, respectively\;
  $Z \leftarrow \sum_{v \in V} \frac{1}{ d(v)}\left({(X_v-Y_v)}^2-X_v-Y_v\right)$\;
  \Return $\frac{Z}{r^2}$.
\end{algorithm}

\begin{proof}[Proof of Lemma~\ref{lemma:dist}]
  The pseudocode of our algorithm is given in Algorithm~\ref{alg:difference-test}.
  Note that the probability that we abort at Line~\ref{line:difference-test-abort} is at most $1/8$.
  Now, we show that that a variant of Algorithm~\ref{alg:difference-test} that does not abort even when $k > 8r$ outputs an estimate of $\|(\p-\q)D^{-1/2}\|_2^2$ with additive error of $\xi$ with probability at least $7/8$.
  Then the claim follows by a union bound.

  For a vertex $v \in V$, let $Z_v:=({(X_v-Y_v)}^2-X_v-Y_v)/ d(v)$.
  Note that $X_v$ is distributed as the Possion distribution $\Poi(r \p(v))$.
  Thus,
  \begin{align*}
    \E[Z_v]
    & =\E\left[\frac{1}{ d(v)}\left({(X_v-Y_v)}^2-X_v-Y_v\right) \right] \\
    &= \frac{1}{ d(v)}\left(\E[X_v^2]-2\E[X_v]\cdot \E[Y_v]+\E[Y_v^2]-\E[X_v]-\E[Y_v]\right) \\
    & =\frac{r^2}{ d(v)}{(\p(v)-\q(v))}^2.
  \end{align*}
  This further implies that
  \[
    \E[Z]=r^2\sum_{v \in V} \frac{1}{d(v)}{(\p(v)-\q(v))}^2=r^2\|D^{-\frac12}(\p-\q)\|_2^2.
  \]
  Now we calculate the variance of $Z$.
  First, we have
  \[
    \Var[Z_v]=\frac{4}{ {d(v)}^2}\left({(\p(v)-\q(v))}^2(\p(v)+\q(v))r^3+2{(\p(v)+\q(v))}^2r^2\right).
  \]
  Thus, we have
  \[
    \Var[Z]=\sum_{v\in V} \frac{4}{{d(v)}^2} \left({(\p(v)-\q(v))}^2(\p(v)+\q(v))r^3+2{(\p(v)+\q(v))}^2r^2\right).
  \]
  Since $\|\p D^{-1/2}\|_2^2=\sum_{v\in V}\frac{\p(v)^2}{d(v)}\leq b$, $\|\q D^{-1/2}\|_2^2=\sum_{v\in V}\frac{\q(v)^2}{d(v)}\leq b$,  we have $\sum_{v\in V}{(\p(v)+\q(v))}^2 / d(v) \leq 4b$ and thus
  \begin{align*}
    \sum_{v\in V}\frac{1}{d^{3/2}(u)} {(\p(v)-\q(v))}^2(\p(v)+\q(v))
    & \leq \sqrt{\sum_{v\in V}\frac{1}{ {d(v)}^2}{(\p(v)-\q(v))}^4}\cdot \sqrt{\sum_{v \in V} \frac{1}{ d(v)}{(\p(v)+\q(v))}^2}  \\
    & \leq 2\sqrt{b} \norm{(\p-\q)D^{-\frac12}}_4^2\leq 2\sqrt{b} \|(\p-\q)D^{-\frac12}\|_2^2,
  \end{align*}
  where the last inequality follows from the fact that $\norm{\x}_4^2\leq \norm{\x}_2^2$ for any vector $\x$. 
  Therefore, by the assumption that the minimum degree is at least one, we have 
  \begin{align*}
    \Var[Z]&=\sum_{v\in V} \frac{4}{{d(v)}^{1/2}}\cdot \frac{1}{d(v)^{3/2}} \left({(\p(v)-\q(v))}^2(\p(v)+\q(v))r^3\right)+\sum_{v\in V} \frac{8}{{d(v)}^{2}}\left({(\p(v)+\q(v))}^2r^2\right)\\
    &\leq 4\cdot 2\sqrt{b} \norm{(\p-\q)D^{-\frac12}}_2^2\cdot r^3+16r^2\sum_{v\in V} \frac{1}{d(v)}\left(\p(v)^2+\q(v)^2\right)\\
    &\leq 8r^3\norm{(\p-\q)D^{-\frac12}}_2^2\sqrt{b}+32r^2\sqrt{b}.
  \end{align*}

For notation simplicity, we let $x=\norm{(\p-\q)D^{-\frac12}}_2^2$ and thus $\Var[Z]\leq 8r^3x\sqrt{b}+32r^2\sqrt{b}$. Then by Chebyshev's inequality, we have that 
 \begin{align*}
&\Pr\left[\abs{\frac{Z}{r^2}-x}>\xi+x\right]
=\Pr\left[\abs{\frac{Z}{r^2}-\E\left[\frac{Z}{r^2}\right]}>\xi+x\right]\\
&\leq \frac{\Var[Z]}{{(\xi+ x)}^2r^4} \leq \frac{8r^3x\sqrt{b}+32r^2\sqrt{b}}{{(\xi+ x)}^2r^4}=\frac{8rx\sqrt{b}+32\sqrt{b}}{{(\xi+ x)}^2r^2}\leq \frac{1}{8},
 \end{align*}
where the last inequality follows from our setting that $r\geq \Theta(\frac{\sqrt{b}}{\xi^2})$ and that $x\leq c$ for some constant $c>0$. 

Thus, our estimator approximates $\|(\p-\q)D^{-1/2}\|_2^2$ within an additive error $\xi'=\frac{\xi}{2}$ with probability at least $7/8$. This concludes the proof of the lemma. 
\end{proof}

\end{document}